\documentclass[twoside]{article}

\usepackage[accepted]{aistats2024}
\usepackage[round]{natbib}

\usepackage{amssymb}
\usepackage{amsfonts}
\usepackage{amsmath}
\usepackage{amsthm}
\usepackage{float}
\usepackage{algorithm}
\usepackage{algpseudocode}
\usepackage{algcompatible}
\usepackage{hyperref}
\hypersetup{colorlinks,linkcolor={blue},citecolor={red},urlcolor={blue}}
\usepackage{graphicx}
\usepackage{xcolor}
\usepackage[capitalise]{cleveref}
\usepackage{tikz}
\usepackage{booktabs}
\usepackage[most]{tcolorbox}
\usepackage{subcaption}
\usepackage{xr}
\makeatletter
\newcommand*{\addFileDependency}[1]{
  \typeout{(#1)}
  \@addtofilelist{#1}
  \IfFileExists{#1}{}{\typeout{No file #1.}}
}
\makeatother

\newcommand*{\myexternaldocument}[1]{%
    \externaldocument{#1}%
    \addFileDependency{#1.tex}%
    \addFileDependency{#1.aux}%
}
\myexternaldocument{supplement}

\newtheorem{proposition}{Proposition}

\newtheorem{definition}{\textbf{Definition}}

\DeclareMathOperator*{\argmin}{arg\,min}

\DeclareMathOperator*{\argmax}{arg\,max}

\newcommand\iidsim{\stackrel{\text{i.i.d.}}{\sim}}

\newcommand{\ahtis}{{\texttt{AHTIS}}}

\definecolor{dorange}{RGB}{223,116,0}
\definecolor{dred}{RGB}{170,0,0}
\definecolor{dgreen}{RGB}{5,118,0}

\begin{document}

%

%

\twocolumn[

\aistatstitle{Adaptive importance sampling for heavy-tailed distributions via $\alpha$-divergence minimization}

\aistatsauthor{ Thomas Guilmeau$^{\diamond}$\footnotemark[1] \And Nicola Branchini$^{\diamond}$\footnotemark[2] \And  Emilie Chouzenoux\footnotemark[1] \And Víctor Elvira\footnotemark[2] }

\aistatsaddress{\footnotemark[1] Université Paris-Saclay, CentraleSupélec, INRIA, CVN, France   \\ \footnotemark[2] University of Edinburgh, United Kingdom  } ]

\begin{abstract}
Adaptive importance sampling (AIS) algorithms are widely used to approximate expectations with respect to complicated target probability distributions. When the target has heavy tails, existing AIS algorithms can provide inconsistent estimators or exhibit slow convergence, as they often neglect the target’s tail behaviour. To avoid this pitfall, we propose an AIS algorithm that approximates the target by Student-t proposal distributions. We adapt location and scale parameters by matching the escort moments - which are defined even for heavy-tailed distributions - of the target and proposal. These updates minimize the $\alpha$-divergence between the target and the proposal, thereby connecting with variational inference. We then show that the $\alpha$-divergence can be approximated by a generalized notion of effective sample size and leverage this new perspective to adapt the tail parameter with Bayesian optimization. We demonstrate the efficacy of our approach through applications to synthetic targets and a Bayesian Student-t regression task on a real example with clinical trial data.

\end{abstract}

\section{INTRODUCTION}
Expectations that are challenging to compute arise repeatedly in probabilistic machine learning \citep{ghahramani2015probabilistic}, Bayesian statistics \citep{robert2007bayesian}, statistical signal processing \citep{sarkka2023bayesian}, option pricing in mathematical finance \citep{l2004quasi}, and many other fields where Monte Carlo methods are often the de-facto standard. Importance sampling (IS) generalizes the Monte Carlo integration principle to approximate expectations with respect to a target distribution $\pi$ \citep{robert1999monte,mcbook,kroese2014monte}. In IS, samples are obtained from a distribution $q$ called proposal that is not necessarily equal to $\pi$. 

Constructing an adequate proposal $q$ is difficult yet crucial for the performance of IS. Adaptive IS (AIS) algorithms, which iteratively refine the proposal distributions, have become the standard to construct efficient samplers~\citep{bugallo2017adaptive}. AIS proposal adaptation procedures can be based on moment matching \citep{cornuet2012}, gradient updates \citep{elvira2015gradient,elvira2023,elvira:hal-03136318}, or combined with Markov Chain Monte Carlo \citep{botev2013markov, martino2017layered,thin2021neo}. 

Several recent works have also highlighted connections between AIS and variational inference (VI) \citep{domke2018importance,finke2019importance,dhaka2021challenges,mattei2022uphill,zhang2022pathfinder,kviman2022multiple}, a framework popular in Bayesian statistics, machine learning and signal/image processing \citep{jordan1999introduction,blei2017variational,marnissi:hal-01613200}. Indeed, VI methods aim at approximating a target $\pi$ with a distribution $q$, by explicitly minimizing a statistical divergence, typically the Kullback-Leibler (KL) divergence. In IS and AIS, the most widely used criterion to evaluate performance is the effective sample size (ESS), which has some connections with a statistical divergence.

In this paper, we focus on a class of AIS procedures based on moment matching, in particular on the AMIS framework of \cite{cornuet2012} that is behind recent state-of-the-art AIS algorithms \citep{paananen2021implicitly}. Although popular, moment-matching updates can be ill-defined when the target or the proposal is heavy-tailed with undefined moments. Notable applications with heavy-tailed $\pi$  include: Student-t error models in Bayesian regression, realistic posterior distributions that are robust to outliers or promoting sparse solutions \citep{fernandez1998bayesian,tipping2005,Amrouche2022}; applied econometrics, where parameter estimation for stochastic volatility models of option pricing involves complicated heavy-tailed distributions \citep{chib2002markov}; analysing financial returns datasets \citep{roy2010monte}. Similarly, heavy-tailed proposals $q$ can be beneficial in AIS \citep[Chapter 9]{mcbook}, although they may not have finite moments thus preventing the application of existing moment-matching methods.

\textbf{Contributions.} \textbf{(1)} We propose an AIS framework, hereby named \ahtis{} (adaptive heavy-tailed importance sampling), allowing heavy-tailed target and proposal distributions. Its proposal adaptation mechanism is based on matching the moments of \emph{escort} densities associated to the target and proposal, i.e., versions of the density with lighter tails. \textbf{(2)} We show that our proposed moment matching corresponds to the minimization of an $\alpha$-divergence. Our approach generalizes previous AIS moment-matching procedures restricted to the KL divergence. 
\textbf{(3)} We show that a generalized notion of effective sample size, the $\alpha$-ESS, is an IS approximation of the $\alpha$-divergence, providing new connections between VI and AIS.
\textbf{(4)} Finally, we exploit this new insight to design a new joint adaptation strategy for the tail and the location/scale parameters of the proposals using Bayesian optimization, outperforming existing moment-matching AIS both when a good tail parameter is known in advance and when it needs to be adapted. This advantage of \ahtis{} is illustrated in \cref{fig:1d-escortVSplain}.
\begin{figure}[H]
    \centering
    \includegraphics[width=0.8\columnwidth]{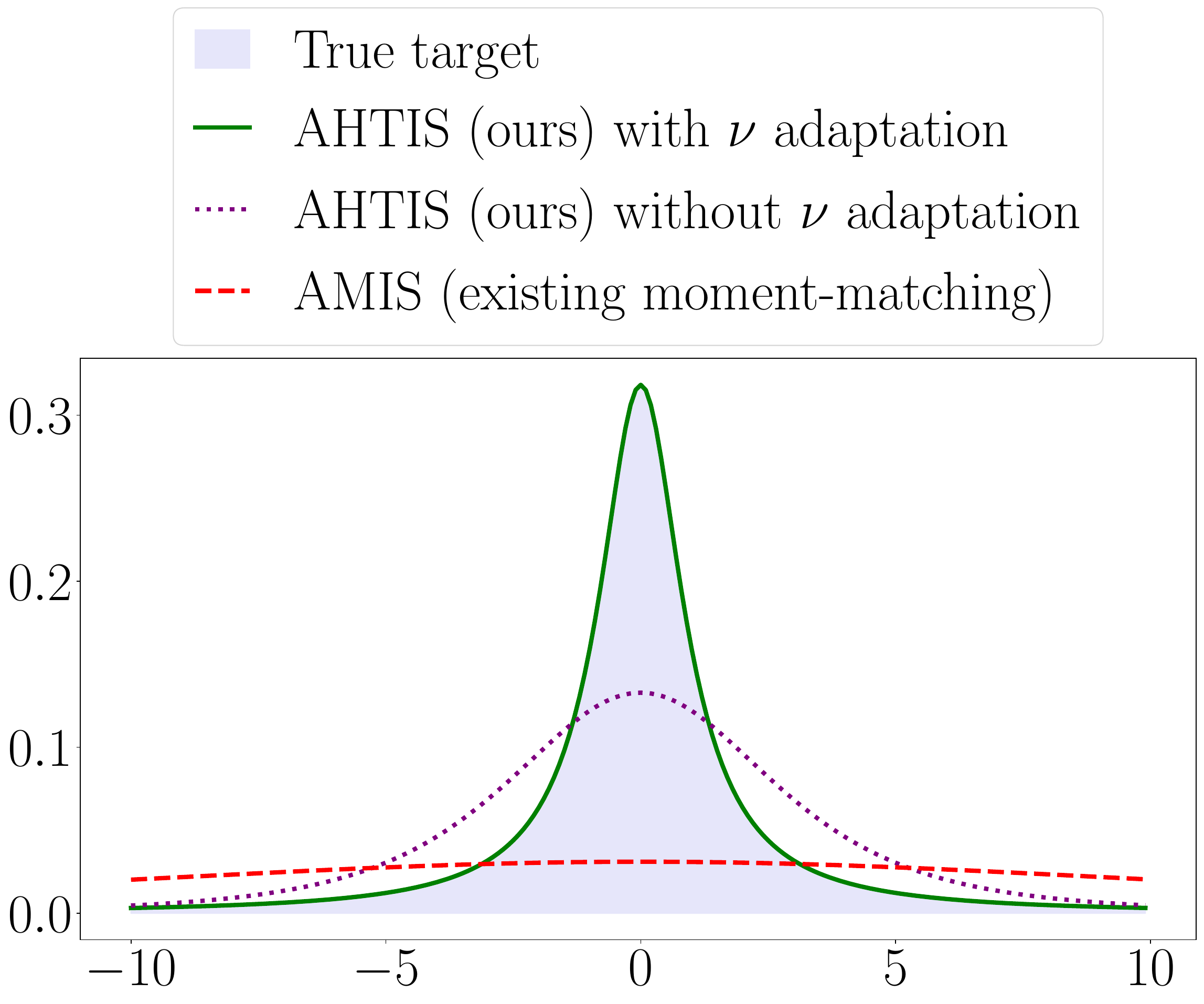}
    \caption{In this illustrative example, the target $\pi$ is a Student-t distribution with $\nu_{\pi}=1$ degrees of freedom (a Cauchy distribution), which is very heavy-tailed and has undefined mean and variance. We show three Student-t approximations using \textbf{(i)} existing moment matching with $\nu=3$ degrees of freedom ($\nu> 2$ is required for the proposals to have moments), \textbf{(ii)} \ahtis{} with $\nu=3$ degrees of freedom, and \textbf{(iii)} \ahtis{} with degrees of freedom adaptation.}
    \label{fig:1d-escortVSplain}
\end{figure}

\paragraph{Outline.}
In Section \ref{section:background}, we introduce the notion of escort probability and our approximating family. In Section \ref{section:method}, we introduce our AIS algorithm, \ahtis, with adaptation of the location, scale, and tail parameters of its proposal. Finally, we show the performance of \ahtis{} on heavy-tailed target distributions in Section \ref{section:experiments}, before concluding in Section \ref{section:conclusion}.

\section{BACKGROUND}
\label{section:background}

\subsection{Importance Sampling}

Importance sampling allows for the Monte Carlo integration of integrals of the form $I = \int h(x) \pi(x)dx$ when samples from the target density $\pi$ are either unavailable, or even inefficient (such as in rare events). Instead, one samples from a proposal distribution $q$ and uses so-called importance weights to correct the estimation. The simplest IS estimator of $I$ is the unnormalized IS estimator (UIS), given by  
\begin{equation}
 \widehat{I}_{\text{UIS}} =  \sum_{m=1}^{M} w^{(m)} h(x^{(m)}) , ~ \{ x^{(m)} \}_{m=1}^{M} \iidsim q(x), 
\end{equation}
where $w^{(m)} = \pi(x^{(m)}) / q(x^{(m)}) $ are the (unnormalized) importance weights using the target probability density function (pdf), $\pi(x)$. When $q = \pi$, $\widehat{I}_{\text{UIS}}$ recovers the plain Monte Carlo estimator, $\widehat{I}_{\text{MC}}$. In many cases, we only have access to the unnormalized target density $\widetilde{\pi}(x) = \pi(x) Z_{\pi}$, i.e., the normalizing constant $Z_{\pi}$ is unknown. The standard estimator for $Z_{\pi}$ is 
\begin{align} \label{eq:z_estimate}
    Z_{\pi} \approx \widehat{Z}_{\pi} = \frac{1}{M} \sum_{m=1}^{M} \frac{\widetilde{\pi}(x^{(m)})}{q(x^{(m)})} = \frac{1}{M} \sum_{m=1}^{M} \widetilde{w}^{(m)} .
\end{align}
\cref{eq:z_estimate} allows one to estimate $I$ when $Z_{\pi}$ is unknown, leading to the self-normalized IS (SNIS) estimator
\begin{equation}
    I \approx \widehat{I}_{\text{SNIS}} =  \sum_{m=1}^{M} \overline{w}^{(m)} h(x^{(m)}) ,
\end{equation}
where $\overline{w}^{(m)} = \widetilde{w}^{(m)} / \sum_{\ell=1}^{M} \widetilde{w}^{(\ell)}$. The almost sure convergence $\widehat{I}_{\text{SNIS}} \xrightarrow[M \rightarrow +\infty]{a.s.} I$ is guaranteed as soon as $\pi(x) > 0 \Rightarrow q(x) > 0$ \citep[Chapter 4]{mcbook}.  

\textbf{Assessing IS performance.} The mean-squared error (MSE) is a common way to evaluate the performance of estimators \citep{mcbook} and, for both $\widehat{I}_{\text{UIS}}$ and $\widehat{I}_{\text{SNIS}}$, the MSE decays at the standard Monte Carlo rate $\mathcal{O}(1/M)$. See, e.g., \citep[Chapter 8]{chopin2020introduction} for more theoretical properties of IS estimators. However, the MSE is usually difficult to evaluate in AIS algorithms. The effective sample size (ESS) is a more practical, and widely used, metric to assess the quality of IS estimators on the fly. It is a sample approximation of the ratio of variances between the SNIS estimator and a Monte Carlo estimator with $\pi$ \citep{kong1992note, elvira2022rethinking}, computed as 
\begin{equation}\label{eq:ess}
    \widehat{\mathrm{ESS}} = \frac{1}{\sum_{m=1}^M \left(\overline{w}^{(m)}\right)^2} \approx \mathrm{ESS} = \frac{\mathbb{V}_{q}[\widehat{I}_{\text{SNIS}}]}{\mathbb{V}_{\pi}[\widehat{I}_{\text{MC}}]} .
\end{equation}
While the original motivation is the above approximation of a ratio of variances, $\widehat{\mathrm{ESS}}$ has been shown to be connected with the chi-squared divergence $\chi^2(\pi,q)$ \citep{orsak1991constrained,agapiou2017importance,sanz2018importance,sanz2020bayesian,akyildiz2021convergence, agarwal2022principled,elvira2022rethinking}. Therefore, the choice of proposal $q$ is crucial to achieve good performance in the above metrics, which led to the development of adaptive IS algorithms (AIS), where proposals are iteratively adapted \citep{bugallo2017adaptive}. 

\textbf{Adaptive multiple IS (AMIS).} AIS algorithms recycle samples to improve the quality of $\widehat{I}_{\text{SNIS}}$. Suppose we have $T$ proposals $\{q_t\}_{t=1}^T$ and for every $t \in \{1,\dots,T\}$ the samples are $\{x^{(m)}_t \}_{m=1}^{M}$. One way to re-use all the $TM$ samples is to assign to each of them an unnormalized weight $\widetilde{w}_t^{(m)} = \frac{\widetilde{\pi}(x_t^{(m)})}{q_t(x_t^{(m)})}$, and possibly perform a resampling step. It has been shown that an alternative, deterministic mixture (DM) weighting, achieves better results by considering \emph{all} the proposals in the weighting of each sample \citep{elvira2019}. The unnormalized DM weight of the sample $x_t^{(m)}$ reads
\begin{equation}
    \widetilde{w}_t^{(m)} = \frac{\widetilde{\pi}(x_t^{(m)})}{\frac{1}{T} \sum_{\tau=1}^T  q_{\tau}(x_t^{(m)})}.
\end{equation}
The DM weighting is notably used by the adaptive multiple importance sampling (AMIS) algorithm proposed in \cite{cornuet2012}, where at each iteration, the proposal is adapted using all the past samples using DM weights. \cite{cornuet2012} suggest to use the DM weights to update the proposal such that its moments match the (estimated) moments of $\pi$.

\subsection{Escort Distributions and $\alpha$-Divergence Minimization}

We introduce now existing results about the minimization of statistical divergences over Student-t distributions which we will use to develop our new method.

\begin{definition}[Multivariate Student-t]
    The multivariate Student-t distribution on $\mathbb{R}^d$ with $\nu> 0$ degrees of freedom, location parameter $\mu \in \mathbb{R}^d$, and positive-definite scale matrix $\Sigma \in \mathcal{S}_{++}^d$ has a pdf with respect to the Lebesgue measure of the form
    \begin{equation}
        q_{\mu, \Sigma, \nu}(x) \propto \left( 1 + \frac{1}{\nu} (x-\mu)^{\top}\Sigma^{-1}(x-\mu) \right)^{- \frac{\nu + d}{2}}
    \end{equation}
    and is normalized by $Z_{\nu, \Sigma} = \frac{\Gamma(\frac{\nu}{2}) }{\Gamma(\frac{\nu+d}{2})} \left(\nu^{d} \pi^{d}\det(\Sigma)\right)^{\frac{1}{2}}$.
\end{definition}

Student-t distributions recover Cauchy distributions when $\nu = 1$ and Gaussian distributions in the limit $\nu \rightarrow + \infty$. They have finite first moment for $\nu > 1$ and finite second moment for $\nu > 2$. Next, we introduce the concepts of escort distribution and $\alpha$-divergence, which will be used throughout \cref{section:method}.

\begin{definition}[Escort version of a pdf]
    \label{def:escortProb}
    Consider $\alpha > 0$ and a pdf $p$. The \emph{escort} version of $p$ \citep{tsallis2009introduction} with exponent $\alpha$ is the pdf $p^{(\alpha)}$ defined by
    \begin{equation}
        p^{(\alpha)}(x) = \frac{p(x)^{\alpha}}{\int p(x)^{\alpha}dx},
    \end{equation}
    assuming that the normalizing constant is finite.
\end{definition}

\begin{definition}[$\alpha$-divergence]
    \label{def:alphadiv}
The $\alpha$-divergence is defined for $\alpha > 0$ and $\alpha \neq 1$ as
\begin{equation}
    \mathrm{D}_{\alpha}(\pi,q) = \frac{1}{\alpha (\alpha-1)} \left (\int \pi(x)^\alpha q(x)^{1- \alpha} dx - 1 \right).
\end{equation}
Its discrete counterpart $\mathrm{D}^M_{\alpha}(\cdot, \cdot)$ is defined similarly on the simplex of $\mathbb{R}^M$, denoted by $\Delta_M$.
\end{definition}
The $\alpha$-divergence generalizes many well-known divergences such as $\mathrm{KL}(\pi, q)$ ($\alpha\rightarrow 1$) and $\chi^2(\pi, q) (\alpha=2)$. The KL divergence is such that $\theta \longmapsto \mathrm{KL}(\pi,q_{\theta})$ is minimized under a moment-matching property when the pdf $q_{\theta}$ form an exponential family \cite[Equation (10.187)]{bishop2006}. This is the case of Gaussian distributions and hence, the optimal KL approximation of $\pi$ is the Gaussian pdf with same first and second order moments as $\pi$. The above result has been generalized beyond this setting, as the next result shows.

\begin{proposition}
    \label{prop:escortMM} \citep{MLE-VI-lambdaExpFamilies}
    Consider a target pdf $\pi$ and the family of Student-t distributions with $\nu> 0$ degrees of freedom. If the escort pdf $\pi^{(\alpha)}$ with $\alpha = 1 + \frac{2}{\nu +d}$ exists and has finite first and second-order moments, then the parameters $(\mu^{\star}_{\nu}, \Sigma^{\star}_{\nu})$ such that
    \begin{equation}
        \label{eq:optCondStudent}
        \begin{cases}
            \mu^{\star}_{\nu} = \pi^{(\alpha)}(x)\\
            \Sigma^{\star}_{\nu} = \pi^{(\alpha)}(x x^{\top}) -\mu_{\nu}^{\star} \mu_{\nu}^{\star \top}
        \end{cases}
    \end{equation}
    minimize $(\mu,\Sigma) \longmapsto \mathrm{D}_{\alpha}(\pi, q_{\mu,\Sigma,\nu})$ .
\end{proposition}

\section{ADAPTIVE HEAVY-TAILED IMPORTANCE SAMPLING}
\label{section:method}
We now present our proposed AIS framework, \ahtis, for handling target distributions with heavy tails and potentially undefined moments based on $\alpha$-divergence minimization. Our framework is summarized in \cref{alg:ours}, which we describe next. In \cref{sec:alphaess}, we show that the so-called $\alpha$-ESS can be used to approximate the $\alpha$-divergence. We exploit this insight to propose our tail parameter $\nu$ adaptation in \cref{sec:adaptingtail}.

\subsection{Step-by-step Breakdown of \ahtis{} and Justification.} 
First, as input to \cref{alg:ours} we require initial location, scale, and tail parameters for the proposal, i.e., $(\mu_0, \Sigma_0, \nu_0)$ respectively. The algorithm follows the following steps for $T>0$ iterations. First, we generate samples from $q_{\mu_t,\nu_t,\Sigma_t}$ (\textbf{step 2}). Then, tail adaptation (\textbf{step 3}) finds $\nu_{t+1}$ (and therefore $\alpha_{t+1})$ with Bayesian optimization, which we detail in \cref{sec:adaptingtail}. The weighting (\textbf{step 4}) uses the DM approach described in \cref{section:background} allowing the proposal to learn from \emph{all} the generated samples. Note that the numerator involves the escort version of the target, $\pi^{(\alpha_{t+1})}$. Notably, this means that when the variance of the weight with respect to the true target $\pi$ is infinite (as it would be the case for existing AIS algorithms, and is common), since $\alpha_t > 1$, the variance of \cref{alg:weighting} may still be finite. 
 
Finally, the escort moment-matching (\textbf{step 5}) minimizes $(\mu,\Sigma) \longmapsto \mathrm{D}_{\alpha_{t+1}}(\pi,q_{\mu,\Sigma,\nu_{t+1}})$ as explained in \cref{section:background}.
\begin{algorithm}[htb]
    \caption{\ahtis}
    \label{alg:ours}
    \begin{algorithmic}[1]
    
        \REQUIRE $\nu_0 > 0$, $\mu_0 \in \mathbb{R}^d$, $\Sigma_0 \in \mathcal{S}_{++}^d$  \vspace{2pt}
        
        \FOR{$t=0,..., T$}\vspace{2pt}
        
            \STATE \textbf{Sampling:} $x_{t}^{(m)} \iidsim q_{\mu_{t}, \Sigma_{t}, \nu_{t}}$, $m = 1,\dots,M$.\vspace{2pt}
            
            \STATE {\bfseries Tail adaptation with BO:} \begin{itemize}
                \item  If $t=0$, $\nu_{1}=\nu_0$, else, set $\nu_{t+1}$ with \cref{alg:tailadapt} in Appendix \ref{appendix:tailAdaptation}.
                \item Set $\alpha_{t+1} = 1 + \frac{2}{\nu_{t+1} + d}$.
            \end{itemize}  \vspace{2pt}

            \STATE {\bfseries Temporal DM weighting:} For $m = 1,\dots,M$ and $\tau = 0,\dots,t$, compute the unnormalized importance weights using the unnormalized escort target as 
            \begin{equation}\label{alg:weighting}
                \widetilde{w}_{\tau}^{(m)} = \frac{\left(\widetilde{\pi}(x_{\tau}^{(m)}) \right)^{\alpha_{t+1}}}{\frac{1}{t+1} \sum_{k=0}^{t} q_{\mu_{k}, \Sigma_{k}, \nu_{k}} (x_{\tau}^{(m)})}
            \end{equation} 
            and normalize to obtain $\overline{w}_{\tau}^{(m)} = \widetilde{w}_{\tau}^{(m)} / \sum_{\ell=1}^{M} \widetilde{w}_{\tau}^{(\ell)}$.\vspace{2pt}

                \STATE {\bfseries Escort moment matching:} Set $(\mu_{t+1}, \Sigma_{t+1})$ with the updates 
            \begin{align}
                \mu_{t+1} &= \sum_{\tau=0}^{t} \sum_{m=1}^{M} \overline{w}_{\tau}^{(m)} x_{\tau}^{(m)}\label{eq:approxOptCondStudent1}\\
                \Sigma_{t+1} &= \sum_{\tau=0}^{t} \sum_{m=1}^{M} \overline{w}_{\tau}^{(m)} x_{\tau}^{(m)}x_{\tau}^{(m)\top} - \mu_{t+1} \mu_{t+1}^{\top}\label{eq:approxOptCondStudent2}
            \end{align}
        \ENDFOR
        \STATE \textbf{Return:} $\{ \overline{w}_{t}^{(m)}, x_{t}^{(m)} \}_{t=1,m=1}^{T,M}$.
    \end{algorithmic}
\end{algorithm}
  \ahtis{} is motivated by the minimization of the $\alpha$-divergence $\mathrm{D}_{\alpha}(\pi,q)$ between target and proposal, which is known to exhibit favourable properties for heavy-tailed distributions \citep{birrell2021variational}, as well as for robust approximate inference with generalized VI on misspecified models in Bayesian statistics \citep{knoblauch2022optimization}. More precisely, \cref{alg:ours} addresses the following joint optimization problem involving $(\mu,\Sigma,\nu)$,
\begin{equation}
    \label{eq:jointMinimization1}
  \mu^{\star}, \Sigma^{\star}, \nu^{\star} 
 = \argmin_{\mu,\Sigma,\nu} \mathrm{D}_{\alpha(\nu)}(\pi,q_{\mu,\Sigma,\nu}).
\end{equation}

Recall from \cref{section:background} that the value $\alpha(\nu)$ in \eqref{eq:jointMinimization1} is such that $\alpha(\nu) = 1 + \frac{2}{\nu+d}$, where $d$ is the dimension of $x$. Hence, we are not minimizing a fixed $\alpha$-divergence, rather jointly adapting the $\alpha$-divergence parameter and the approximating family's degree of freedom parameter $\nu$. We now establish in \cref{eq:wellPosednessProblem} that when $\pi$ is a Student-t distribution, the optimization problem in \cref{eq:jointMinimization1} is solved when the proposal recovers $\pi$, illustrating the rationale of our approach. However, we remark that our algorithm \ahtis{} is not restricted to Student-t targets.
\begin{proposition}[Well-posedness of tail-adaptation] 
    \label{eq:wellPosednessProblem}
    Suppose that the target $\pi$ is a Student-t pdf with $\nu_{\pi} > 0$ degrees of freedom. Then, Problem \eqref{eq:jointMinimization1} is solved by $(\mu^{\star},\Sigma^{\star},\nu^{\star})$ such that $\nu^{\star} = \nu_{\pi}$ and  $q_{\mu^{\star},\Sigma^{\star},\nu^{\star}} = \pi$.
\end{proposition}
The proof is postponed to the Appendix \ref{app:wellPosednessStudentTarget}. To obtain a practical algorithm to minimize the problem in \cref{eq:jointMinimization1}, we propose to consider $(\mu,\Sigma)$ and $\nu$ separately, and equivalently reformulate \cref{eq:jointMinimization1} as 
\begin{equation}
    \label{eq:jointMinimization2}
    \nu^{\star} = \argmin_{\nu} \min_{\mu,\Sigma} \mathrm{D}_{\alpha(\nu)}(\pi,q_{\mu,\Sigma,\nu}).
\end{equation}
This is motivated by the fact that for a given $\nu > 0$, $\min_{\mu,\Sigma} \mathrm{D}_{\alpha(\nu)}(\pi,q_{\mu,\Sigma,\nu}) = \mathrm{D}_{\alpha(\nu)}(\pi, q_{\mu^{\star}_{\nu}, \Sigma^{\star}_{\nu}, \nu})$, with $(\mu^{\star}_{\nu}, \Sigma^{\star}_{\nu})$ satisfying Eq.~\eqref{eq:optCondStudent}. The behaviour of $\nu \longmapsto \mathrm{D}_{\alpha(\nu)}(\pi, q_{\mu^{\star}_{\nu}, \Sigma^{\star}_{\nu}, \nu})$ is illustrated in \cref{fig:optimalRenyiValue} (see Appendix \ref{app:wellPosednessStudentTarget} for details).
\begin{figure}[htb]
    \centering
    \includegraphics[width=\columnwidth]{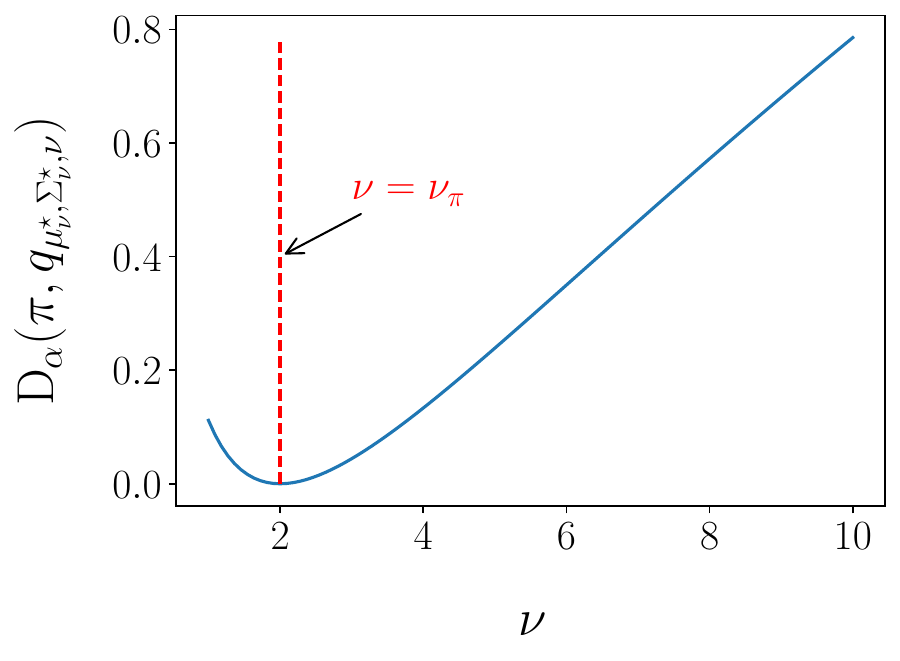}
    \caption{Optimal $\alpha$-divergence value as a function of $\nu > 0$ from Proposition \ref{prop:escortMM} when $\pi$ is a Student-t distribution in dimension $d=5$ and degree of freedom parameter $\nu_{\pi} = 2$ (vertical dotted red line).}
    \label{fig:optimalRenyiValue}
\end{figure}

Next, we propose an approach to solve this optimization problem within \textbf{step (2)} of \cref{alg:ours} without using any additional samples. This requires to evaluate the objective in \cref{eq:jointMinimization2}, which we address now.

\subsection{Connecting VI and IS with the $\alpha$-ESS}\label{sec:alphaess}

A challenge is that, for realistic target distributions $\pi$, one cannot evaluate the cost function $\mathrm{D}_{\alpha(\nu)}(\pi,q_{\mu,\nu,\Sigma})$ appearing in the minimization problem of \cref{eq:jointMinimization2}. We now show that an SNIS approximation of $\mathrm{D}_{\alpha(\nu)}(\pi,q_{\mu,\nu,\Sigma})$ is related in a precise way to an existing generalized ESS, the $\alpha$-ESS, which belongs to the Huggins-Roy family of ESS metrics \citep{martino2017effective,huggins2019sequential}. This result connects further VI and IS and allows us to obtain a practical way to approximate $\mathrm{D}_{\alpha(\nu)}(\pi,q_{\mu,\nu,\Sigma})$, that we will use to adapt the tail parameter $\nu$ in \cref{sec:adaptingtail}.
The $\alpha$-ESS is defined over the simplex $\Delta_M$ as:
\begin{equation}
    \widehat{\mathrm{ESS}}_{\alpha}(\overline{w}) = \left( \sum_{m=1}^M \left(\overline{w}^{(m)}\right)^{\alpha} \right)^{\frac{1}{1-\alpha}},\,\forall \overline{w} \in \Delta_M.
\end{equation}
We now show our main result connecting $\widehat{\mathrm{ESS}}_{\alpha}$ and $\mathrm{D}_{\alpha}(\pi, q)$ for general target and proposal distributions.
\begin{proposition}[Almost sure convergence]
    \label{prop:divergenceESS}
    Consider a target $\pi$ and a proposal $q$ with normalized importance weights $\{ \overline{w}^{(m)}\}_{m=1}^M$ associated with i.i.d.~samples from $q$. Then, the discrete $\alpha$-divergence between the weights $\{\overline{w}^{(m)}\}_{m=1}^M$ and the uniform weights $\{ 1/M \}_{m=1}^M$ is related to $\widehat{\mathrm{ESS}}_{\alpha}$ as follows:
    \begin{align}
    &\mathrm{D}^M_{\alpha}(\{\overline{w}^{(m)}\}_{m=1}^M, \{ 1/M\}_{m=1}^M)  \nonumber \\
    = &\frac{ M^{ \alpha - 1}}{\alpha (\alpha - 1 )}\left( \widehat{\mathrm{ESS}}_{\alpha}(\{ \overline{w}^{(m)}\}_{m=1}^M)^{1 - \alpha} - M^{1 - \alpha} \right).
    \end{align}
    Moreover, $\mathrm{D}_{\alpha}^{M}$ converges to $\mathrm{D}_{\alpha}(\pi,q)$, i.e.,
    \begin{equation}
        \mathrm{D}_{\alpha}^{M}(\{\overline{w}^{(m)}\}_{m=1}^M, \{ 1/M\}_{m=1}^M) \xrightarrow[M \rightarrow +\infty]{a.s.} \mathrm{D}_{\alpha}(\pi,q)
    \end{equation}
    in an almost sure sense when $\pi(x) > 0 \Rightarrow q(x) >0$. 
\end{proposition}
The proof is provided in Appendix \ref{app:proof_ess}. The quantity $\widehat{\mathrm{ESS}}_{\alpha}$ can be cheaply computed. Further, since our derivation shows that $\widehat{\mathrm{ESS}}_{\alpha}$ is specifically a SNIS estimator, we obtain a central limit theorem (CLT) by extending standard SNIS results \citep{chopin2020introduction}, which allows to quantify uncertainty using asymptotic confidence intervals. 
\begin{proposition}[CLT]
\label{prop:clt}
If $\pi(x) > 0 \Rightarrow q(x) > 0$ and $\mathbb{V}_q \left[ \left( \frac{\widetilde{\pi}(x)}{q(x)} \right)^{\alpha} \right] < +\infty$, the estimator $\mathrm{D}_{\alpha}^M$ of the $\alpha$-divergence is $\sqrt{M}$-asymptotically normal, i.e., 
    \begin{align}
    &\sqrt{M} \left( \mathrm{D}_{\alpha}^M(\{\overline{w}^{(m)}\}_{m=1}^M, \{ 1/M\}_{m=1}^M) - \mathrm{D}_{\alpha}(\pi,q) \right) \nonumber \\ 
    &\xrightarrow[N\rightarrow +\infty]{d} \mathcal{N}\left(0, \sigma^2 \right),
\end{align}
with variance
\begin{equation}
    \sigma^2 = \left( \frac{\int \widetilde{\pi}(x)^{2\alpha} q(x)^{1-2\alpha}dx}{\left(\alpha(\alpha-1) \int \widetilde{\pi}(x)^{\alpha}q(x)^{1-\alpha} dx\right)^2} - 1 \right).
\end{equation}
\end{proposition}
See Appendix \ref{app:proof_ess} for a proof. Next, we detail \textbf{step 3} of \cref{alg:ours}, which relies on $\widehat{\mathrm{ESS}}_{\alpha}$.

\subsection{Tail Adaptation with Bayesian Optimization}\label{sec:adaptingtail}

We now describe how to adapt without generating additional samples the parameter $\nu$ within the optimization problem in \cref{eq:jointMinimization1} (the procedure is further detailed in Appendix \ref{appendix:tailAdaptation}).

The outer problem on $\nu$ consists in minimizing the function $\nu \longmapsto \mathrm{D}_{\alpha(\nu)}(\pi, q_{\mu^{\star}_{\nu},\Sigma^{\star}_{\nu},\nu})$, with $(\mu^{\star}_{\nu},\Sigma^{\star}_{\nu})$ satisfying Eq.~\eqref{eq:optCondStudent}. Although one-dimensional, this problem is difficult as it involves intractable integrals and inner optimization. We propose a Bayesian optimization (BO) approach~\citep{garnett_bayesoptbook_2023}. BO algorithms do not require the computations of derivatives and can cope with noisy estimations of the objective function. Further, they only require a small number of these noisy evaluations, which fits well within our context, since in AMIS \citep{cornuet2012}, the value of $T$ does not need to be large (see \cref{section:experiments} for details). 

To solve \eqref{eq:jointMinimization2} with BO, the main challenge is to approximate at every iteration $t =1,\ldots,T$ the quantity $\mathrm{D}_{\alpha(\nu_t)}(\pi, q_{\mu^{\star}_{\nu},\Sigma^{\star}_{\nu},\nu_t})$. To do so, we first remark that 
\begin{equation}
    (\mu_{t}, \Sigma_{t}) \approx \argmin_{\mu,\Sigma}  \mathrm{D}_{\alpha_t}(\pi, q_{\mu,\Sigma,\nu_t}),
\end{equation}
in the sense that $(\mu_{t}, \Sigma_{t})$ are constructed following \eqref{eq:approxOptCondStudent1}-\eqref{eq:approxOptCondStudent2} which are estimators of the optimality conditions \eqref{eq:optCondStudent}. Then, the quantity $\mathrm{D}_{\alpha_t}(\pi, q_{\mu^{\star}_{\nu_t},\Sigma^{\star}_{\nu_t},\nu_t})$ is approximated by computing the $\alpha_t$-ESS with target $\pi$ and proposal $q_{\mu_t,\Sigma_t,\nu_t}$, following our Proposition \ref{prop:divergenceESS}.

BO algorithms construct a probabilistic model of the function $\nu \longmapsto \mathrm{D}_{\alpha(\nu)}(\pi, q_{\mu^{\star}_{\nu},\Sigma^{\star}_{\nu},\nu})$ in the form of a Gaussian process (GP). At every iteration, the GP is updated with the data $\{\nu_{\tau}, \alpha_{\tau}$-$\mathrm{ESS}\}_{\tau=1}^{t}$, where the values $\alpha_{\tau}$-$\mathrm{ESS}$ are seen as noisy observations of the $\alpha_{k}$-divergence. Then, an acquisition function, which governs the trade-off between exploration and exploitation, is maximized, yielding the next value $\nu_{t+1}$. We use Upper Confidence Bound (UCB) as the acquisition function, which offers theoretical guarantees on cumulative regret by balancing exploration and exploitation with a logarithmic regret bound \citep[Chapter 10]{garnett_bayesoptbook_2023}. As kernel for the GP, we use a standard radial-basis function (RBF) kernel with default parameters. For more details on the BO procedure, see Appendix \ref{appendix:tailAdaptation}.

\textbf{Computational complexity of \ahtis{}.} 
The complexity of \cref{alg:ours} can be analysed by the one of AMIS, and the added complexity given by \textbf{step (3)}, the tail adaptation based on BO. Like AMIS, \ahtis{} requires $\mathcal{O}(M T^2)$ proposal evaluations due to the use of deterministic mixture weighting. While this can be prohibitive for large $T$, we find consistent results with the original AMIS paper \citep{cornuet2012} where $T$ does not need to be very large (between $20$ and $30$ in both our examples and theirs) while $M$ is sufficiently large. This implies that the BO procedure (see Appendix \ref{appendix:tailAdaptation}) is not too expensive in practice, even if cubic in $T$ in theory since the GP is fitted on $T$ examples. Note that the dimension of our BO problem is fixed to $1$, since $\nu$ is a scalar. Thus, the complexity of BO is driven by \textbf{(i)} sequentially updating the GP and \textbf{(ii)} maximizing the acquisition function. Many works in the BO literature aim to reduce these costs, see e.g. \citep[Chapters 9.1, 9.2]{garnett_bayesoptbook_2023}. In our case, UCB is one of the cheapest acquisitions to maximize \citep{wilson2018maximizing}. Finally, previous work has also managed to reduce AMIS complexity to $\mathcal{O}(MTK)$ (for some constant $K$) while keeping high efficiency \citep{el2019efficient} whose techniques also apply to \ahtis{}.

\section{RELATED WORKS}\label{section:relatedwork}

In general, AIS methods do not specifically handle heavy-tailed targets with undefined moments. Although some works use heavy-tailed proposals, to the best of our knowledge, no existing AIS work adapts the tail parameter of a heavy-tailed proposal as in \cref{alg:ours}, while some works in VI do so.

\textbf{AIS.} \cite{wang2022moment} in the context of AIS match the first three moments of skew-Student proposals with the target's moments for adaptation, without adapting $\nu$, requiring $\nu > 3$, and with no connection with $\alpha$-divergences. \cite{korba2022adaptive} introduce an AIS scheme using a mixture of an iteratively adapted kernel density estimator and a safe heavy-tailed distribution, however without detailing the latter's construction. Other AIS works using moment matching mention the use of Student-t distributions, but do not adapt the tail parameter $\nu$ \citep{cornuet2012,portier2018asymptotic}.

\textbf{VI.} \cite{daudel2023monotonic} propose a general VI framework that allows in particular to minimize a fixed $\alpha$-divergence over a mixture of Student-t distributions. The location, scale, and tail parameters of the Student-t distributions are adapted. While we adapt $\nu$ using a BO algorithm, they do so by solving a non-linear equation. However, their procedure may not be able to reach low value of $\nu$, contrary to ours (see Appendix \ref{appendix:tailAdaptation} for a justification), and they did not implement a practical scheme showing experimental results . The work of \cite{wang2018} proposes to minimize an $f$-divergence that is implicitly defined at each iteration by the importance weights of the samples. This is connected with the dependence of the $\alpha$-divergence we minimize on the degree of freedom parameter. However, their goal diverges from ours by focusing on obtaining mass-covering proposals. The minimization of an $\alpha$-divergence (or a Rényi divergence) is also considered in \citep{hernandez2016black, li2016renyi}. In these works, the resulting optimization problem is solved by stochastic gradient descent on a general proposal family, while here we exploit the Student-t assumption to obtain direct optimality conditions. Further, note that VI methods \textbf{(i)} do not use recycling of past samples, \textbf{(ii)} usually yield only a lower bound of $Z_{\pi}$. This is in contrast with the AIS literature, where samples recycling strategies such as DM weighting have been used \citep{10.3150/18-BEJ1042}, allowing to construct $\mathcal{O}(1/M)$-consistent estimates of $Z_{\pi}$.

\begin{figure*}[t]
      \centering
      \begin{subfigure}[t]{.49\linewidth}
      \centering
          \makebox[\textwidth]{\includegraphics[width=0.9\linewidth]{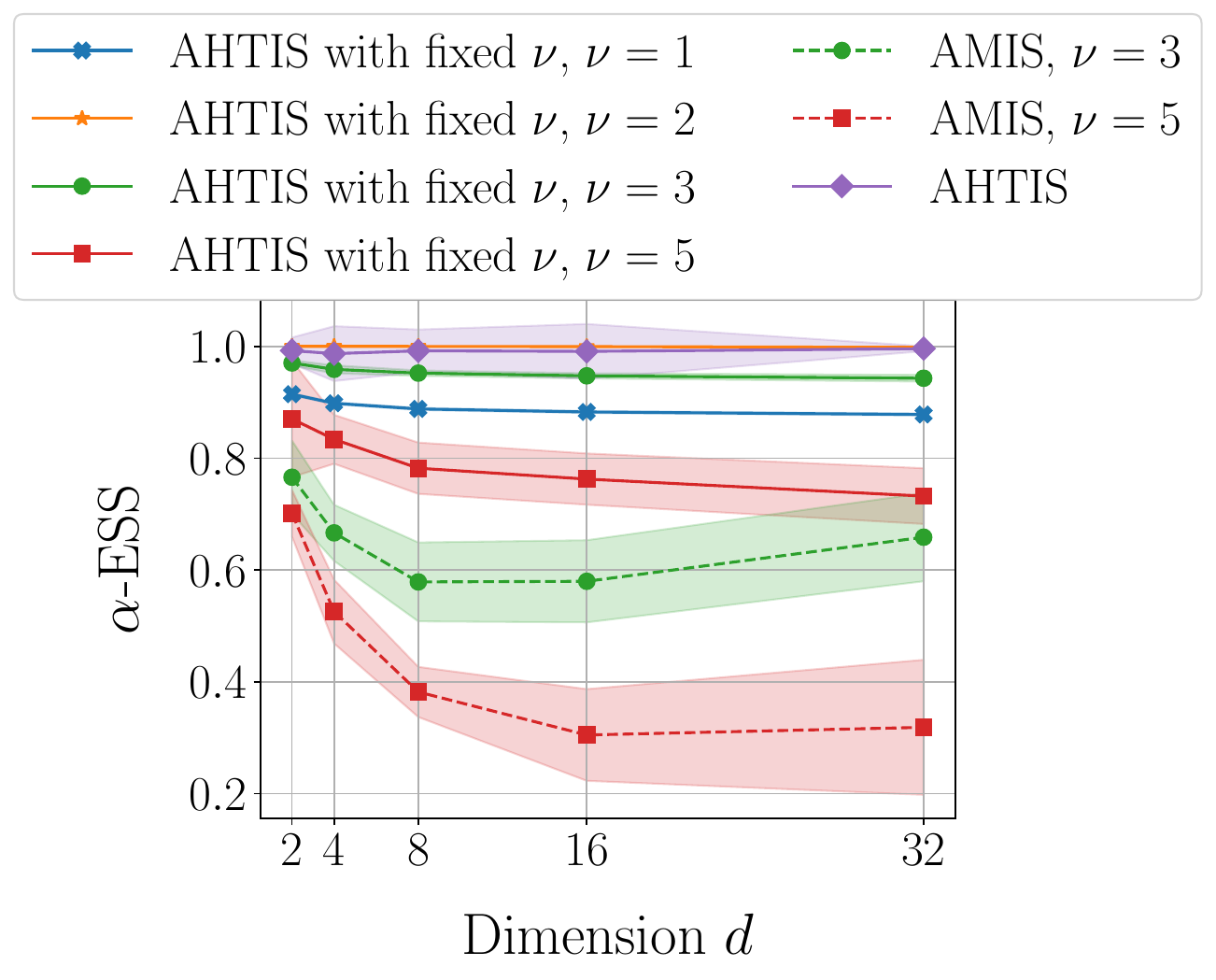}}
            \caption{$\alpha$-ESS (mean $\pm$ one standard deviation, higher is better) for various dimensions $d$. \ahtis{} outperforms AMIS for any $\nu$, sometimes by an order of magnitude, and the $\nu$-adaptive version converges to the true value $\nu_{\pi} = 2$.}
            \label{fig:alphaESS_2dofTarget}
      \end{subfigure} \hfill 
      \begin{subfigure}[t]{.49\linewidth}
      \centering
          \makebox[\textwidth]{\includegraphics[width=0.9\linewidth]{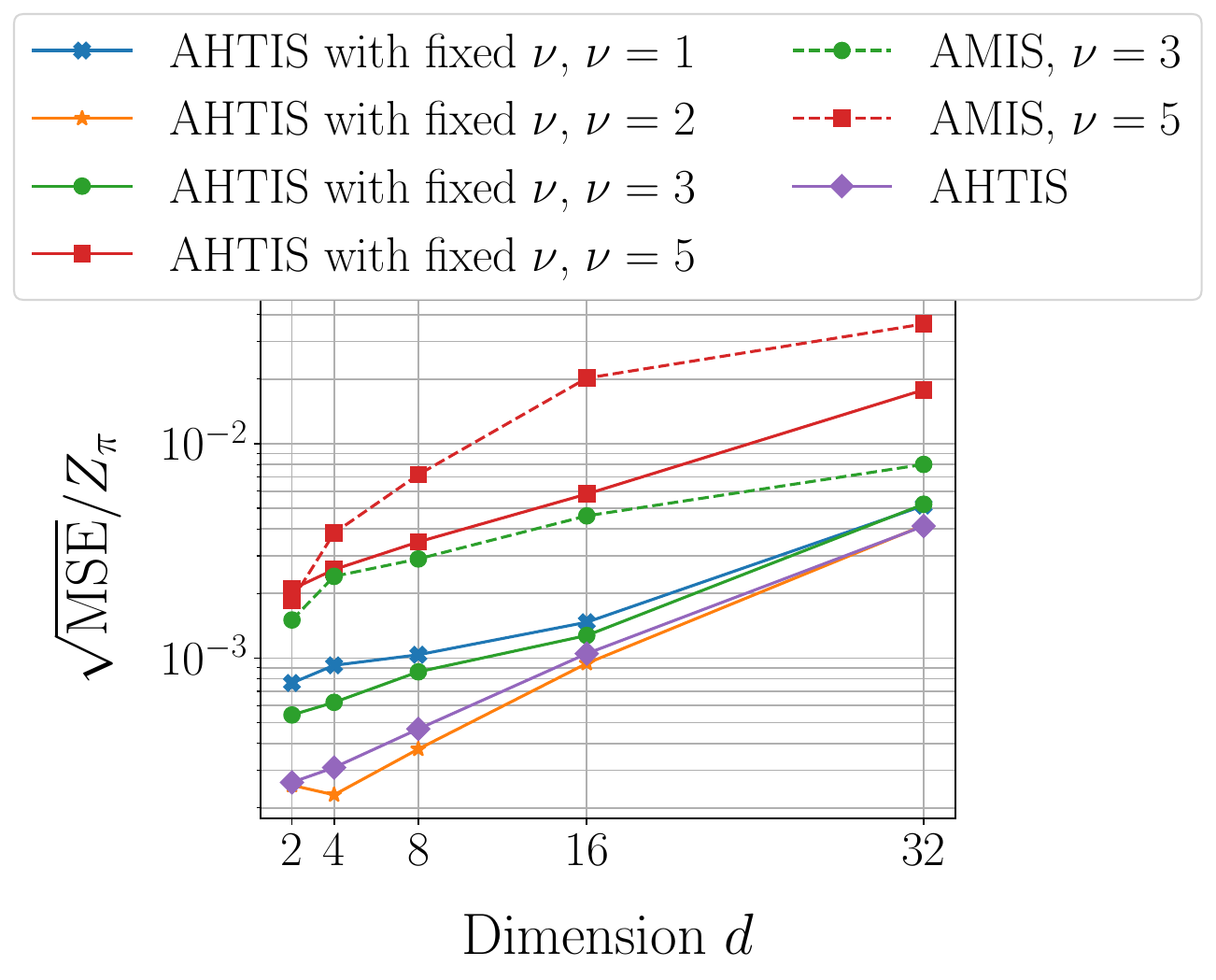}}
                    \caption{Relative square root MSE (lower is better) for various dimensions $d$. Note that $Z_{\pi} = Z_{\nu_{\pi},\Sigma_{\pi}}$ is the true normalizing constant, which is available. \ahtis{} outperforms AMIS for any $\nu$ and the $\nu$-adaptive version converges to $\nu_{\pi} = 2$.}
            \label{fig:MSE_Z_2dofTarget}
      \end{subfigure}
\caption{Results for \cref{subsec:studentTargets}. All algorithms are run for $T = 20$ iterations, with $M=10^4$ samples per iteration and results are averaged over $100$ replications. A dashed line identifies AMIS, while solid line is \ahtis{}, and same marker/color indicates same $\nu$. Recall that AMIS updates are not defined for $\nu \in \{1,2\}$ as the proposal moments are undefined.}
\label{fig:synthetic_res}
\end{figure*}

\section{EXPERIMENTS}
\label{section:experiments}
We demonstrate the benefits of \ahtis{} first on a controlled scenario with synthetic heavy-tailed targets (Student-t distributions of varying dimensions), second on a posterior distribution arising from a Bayesian robust regression problem on clinical trial data.

We evaluate the algorithms using the $\alpha$-ESS metric, shown in \cref{sec:alphaess} to be a theoretically sounded approximation of $\mathrm{D}_{\alpha}$, and the MSE on the estimation of the normalizing constant $Z_{\pi}$, a key distinguishing feature of (A)IS algorithms \citep{llorente2023marginal}.

\subsection{Controlled Scenario with Varying Dimension Student-t Targets}
\label{subsec:studentTargets}
We start with the problem of approximating integrals involving a heavy-tailed Student-t target $\pi$ with $\nu_{\pi} \in \{2,5\}$. Note that the second-order moments of $\pi$ are not defined when $\nu_{\pi} = 2$. The sought target has a location parameter sampled in $\text{Uniform}[-1,1]^d$. Moreover, its scale matrix $\Sigma_{\pi}$ is built so as to reach a condition number $\kappa = 5$, following \citep[Sec. 5]{moré1989}. We consider dimensions $d \in \{2,4,8,16,32\}$. 

We run \ahtis{} and AMIS algorithms for $T=20$ iterations, with $M=10^4$ samples per iteration, following the guidelines from~\citep{cornuet2012}. In the spirit of an ablation study, we analyze the benefits of the tail adaptation in \ahtis{}. That is, we also run \ahtis{} without \textbf{step (3)} of \cref{alg:ours}, $\nu$ being fixed and possibly different from $\nu_{\pi}$.  All algorithms are initialized with $\mu_0$ sampled in $\text{Uniform}[-5,5]^d$ and $\Sigma_0 = 10 I_d$. For \ahtis{} with \textbf{step (3)}, the value $\nu_0 = 1$ is used. Else, the degrees of freedom $\nu \in \{1,2,3,5\}$ are considered for the algorithms without tail adaptation. Note that in the case $\nu\leq 2$ the updates of AMIS are not defined.

\textbf{Results.} The results in terms of the considered metrics are shown in Fig.~\ref{fig:alphaESS_2dofTarget}-\ref{fig:MSE_Z_2dofTarget}. The best performance in both metrics are reached by the $\nu$-adaptive \ahtis{} and by \ahtis{} with $\nu = \nu_{\pi}$. This shows that the $\nu$-adaptive \ahtis{} is able to capture the tail behaviour of the target and confirms the result of \cref{eq:wellPosednessProblem}. When $\nu$ is fixed, \ahtis{} outperforms AMIS in both metrics when $\nu > 2$, and allows in addition to use heavy-tailed proposals with $\nu \leq 2$. Such proposals yield better performance on this heavy-tailed target. We show additional results in Appendix \ref{appendix:furtherNumericalExperiments}, including the case $\nu_{\pi} = 5$ revealing similar behaviours, as well as an analysis of the adaptation of $\nu$ of \ahtis{}.

\begin{figure*}[t]
      \centering
      \begin{subfigure}[t]{.49\linewidth}
      \centering
          \makebox[\textwidth]{\includegraphics[width=\linewidth]{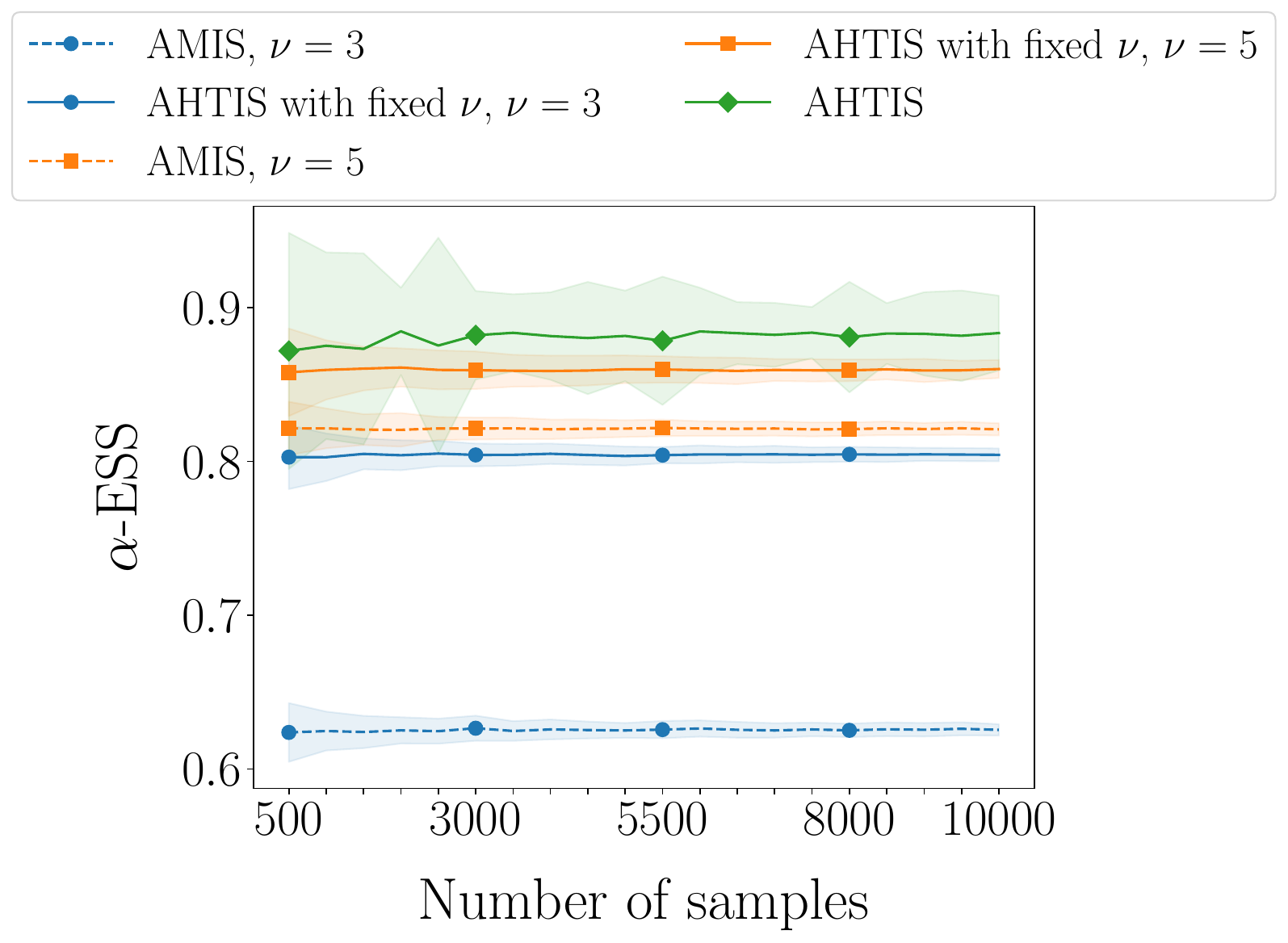}}
            \caption{$\alpha$-ESS (mean $\pm$ one standard deviation, higher is better). \ahtis{} outperforms AMIS with fixed $\nu$. The $\nu$-adaptive \ahtis{} yields better mean, but exhibits a larger variance.}
            \label{fig:alphaESS_real}
      \end{subfigure} \hfill 
      \begin{subfigure}[t]{.49\linewidth}
      \centering
          \makebox[\textwidth]{\includegraphics[width=\linewidth]{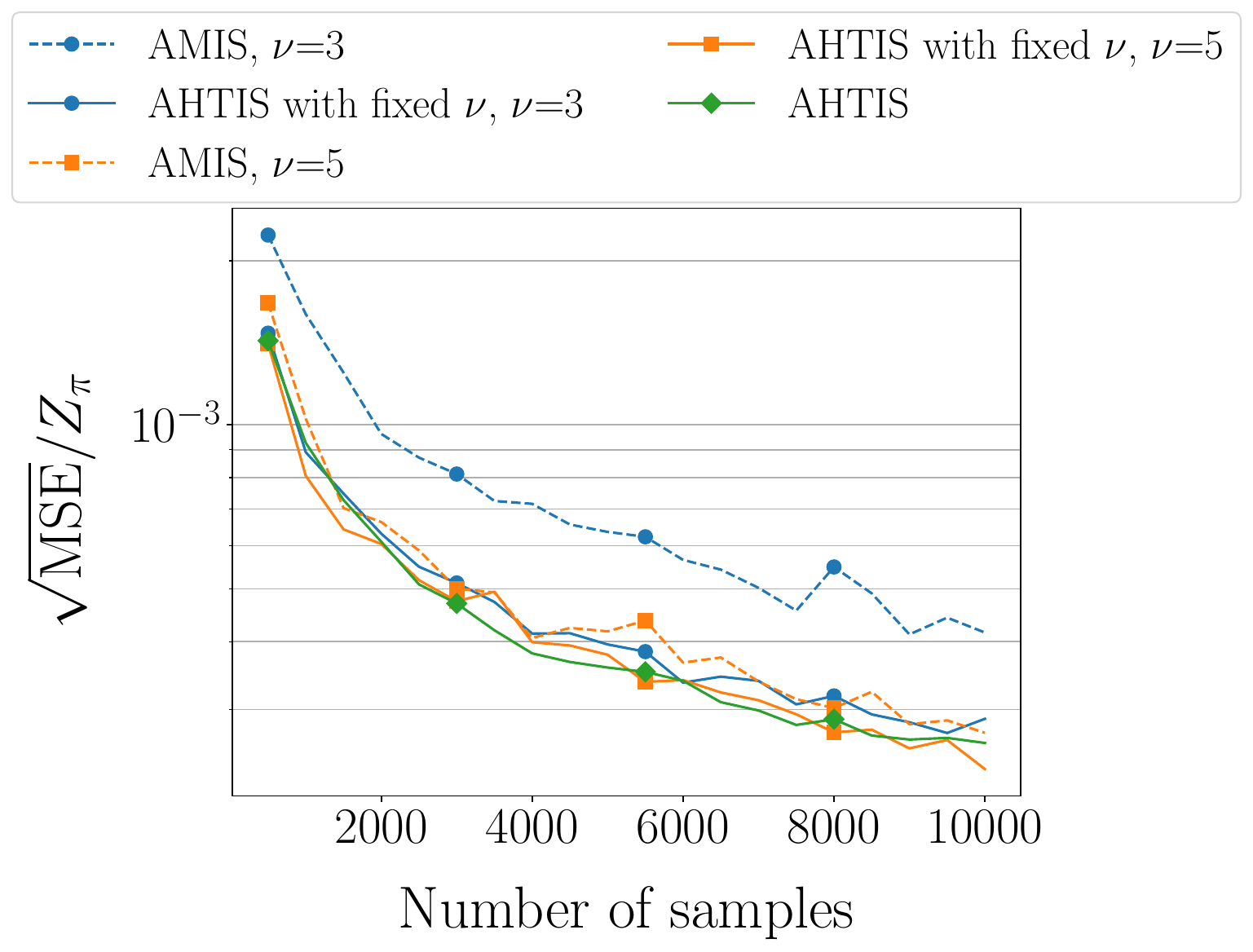}}
                    \caption{In terms of relative square root MSE (lower is better), \ahtis{} outperforms AMIS for any $\nu$ and in this case \ahtis{} reaches the best performance. }
            \label{fig:MSE_Z_real}
      \end{subfigure}
\label{fig:real}
\caption{Results for the creatinine dataset, \cref{exp:real}. All algorithms are run for $T = 25$ iterations, with varying number of samples and results are averaged over $250$ replications. A dashed line identifies AMIS, while solid line is \ahtis{}, and same marker/color indicates same $\nu$.}
\end{figure*}

\subsection{Application to Bayesian Student-t Regression on Real Data}\label{exp:real}
We apply our methodology using, as the target $\pi$, the posterior resulting from a robust regression model on the creatinine dataset \citep{liu1995ml}.\footnote{publicly available at \url{https://github.com/faosorios/heavy/blob/master/data/creatinine.rda}} This dataset has been used to benchmark state-of-the-art VI and MCMC algorithms \citep{DBLP:conf/icml/XuCC23}. It contains the results of a clinical trial on $N=34$ male patients. Such a small number of datapoints makes the inference task challenging, requiring appropriate robust modeling. The regression model assumed in \citep{liu1995ml} to tackle this dataset is a Bayesian Student-t regression for scalar observations $\{y_n\}_{n=1}^N$ representing endogenous creatinine clearance (CR); the covariates $X_n \in \mathbb{R}^{3}$ represent body weight in kg, serum creatinine concentration, and age in years. The goal is to predict CR of the patients. Therefore, the model (which includes an intercept) is given by 
\begin{equation}
    y_n \mid X_n, \beta \iidsim \mathcal{T}([X_n, 1]^{\top} \beta, I_4, 5),
\end{equation}
where $\beta \in \mathbb{R}^{4}$ follows the prior $p_0 = \mathcal{T}(0,I_4,1)$, and $\mathcal{T}(\mu,\Sigma,\nu)$ is the Student-t distribution with location $\mu$, scale $\Sigma$, and $\nu$ degrees of freedom. The posterior pdf $\pi$, with likelihood $p$ and prior pdf $p_0$, is such that
\begin{equation}
    \label{eq:posterior}
    \pi( \beta | \{ X_{n} , y_n \}_{n=1}^{N} ) \propto \left ( \prod_{n=1}^{N}  p(y_{n} | X_{n}, \beta) \right) p_0(\beta).
\end{equation}
The normalization constant of $\pi$, $Z_{\pi}$, is of practical importance as it can be used for model selection and is known as model evidence \citep{mackay1992bayesian}.

We use \ahtis{} and AMIS to approximate the posterior in \eqref{eq:posterior} and approximate $Z_{\pi}$. We use $T=25$ iterations and varying number of samples. In order to obtain a better adaptation of the degree of freedom parameter $\nu$ by \ahtis{} in this case, we optimize the Gaussian process hyperparameters, with regularized maximum likelihood (full details in Appendix \ref{appendix:tailAdaptation}). Algorithms are initialized with $\mu_0$ sampled in $\text{Uniform}[-5,5]^d$, $\Sigma_0 = 4 I_d$ (here, $d=4$). \ahtis{} with adaptation of $\nu$ is initialized with $\nu_0 = 1$ while the algorithms with fixed $\nu$ use $\nu \in \{3,5\}$. 

As before, we evaluate the $\alpha$-ESS and the MSE on the estimation of $Z_{\pi}$. Since we do not have access to the true value of $Z_{\pi}$, we estimate the ground truth using AMIS with $10^5$ samples for $T=25$ iterations and initialised with the Laplace approximation of $\pi$ \citep{mackay1992practical}. Namely, we run AMIS with degree of freedom $\nu = 5$ and initial values $\mu_0 = \argmax_{\beta}  \pi( \beta  | \{ X_{n} , y_n \}_{n=1}^{N} ) $ and $\Sigma_0 = - \left [ \frac{\partial^2 }{\partial \beta^2 } \log \pi( \beta = \mu_0, \{ X_{n} , y_n \}_{n=1}^{N} ) \right ]^{-1} $.

\textbf{Results.} In Fig.~\ref{fig:alphaESS_real}-\ref{fig:MSE_Z_real}, we display the $\alpha$-ESS and the square root relative MSE as functions of the number of samples $M$. 
In this experiment, there is no obvious true value for $\nu$, due to the intractable $\pi$. The $\nu$-adaptive \ahtis{} shows the best mean $\alpha$-ESS values, albeit with a larger variance. 
We expect this to be the case, since \ahtis{} has to learn $\nu$ adaptively with an objective function that, differently to Student-t $\pi$, may exhibit local minimizer and thus may not always reach the global minimizer. In terms of MSE, the best performance is reached by the $\nu$-adaptive \ahtis{}, and second best by \ahtis{} with $\nu=5$ (which motivated using this $\nu$ for the ground truth). Note that when $\nu$ is fixed, \ahtis{} reaches better performance in both metrics than AMIS. We report results with more values of $\nu$ in Appendix \ref{appendix:furtherNumericalExperiments}, with qualitatively similar findings.

\section{CONCLUSIONS}
\label{section:conclusion}
We have proposed \ahtis{}, an AIS framework specifically suited for heavy-tailed target distributions $\pi$, being the first to do so explicitly in the AIS literature. \ahtis{} allows for the adaptation of location, scale, and tail parameter of a Student-t proposal, hereby differing from most previous AIS works. We also explicitly minimize an $\alpha$-divergence between the target and the proposal, in the spirit of VI methods. We showed that the $\alpha$-divergence can be approximated by a quantity involving the $\alpha$-ESS, connecting further AIS and VI algorithms and allowing us to design our tail adaptation method.

Our framework is compatible with the use of mixture proposals when the target is suspected to be multi-modal, and an extension towards this direction is interesting future work. Further, the computational efficiency of the tail adaptation procedure, when a good $\nu$ is not known in advance, could benefit from existing works in the BO literature.

\newpage

\bibliography{references}

\onecolumn 
\appendix 

\section*{Appendix}

In \cref{app:escortExample}, we give an example of the construction of an escort probability density that has lighter tails than the original. In \cref{app:wellPosednessStudentTarget}, we study the well-posedness of our variational formulation of the adaptation of the location, scale, and tail parameters of the proposal. In Section \cref{app:proof_ess}, we give the proofs of our results about the sampling estimation of $\alpha$-divergences. We detail our tail adaptation procedure in \cref{appendix:tailAdaptation}, as well as another tail adaptation procedure proposed recently in the VI literature. Finally, we provide additional numerical experiments in \cref{appendix:furtherNumericalExperiments}.

We run the synthetic experiments on a personal laptop with $7,6$ GB RAM and with 8 Intel Core i$5-8265$U cores. We run the real data experiments on a a personal laptop (MacBook Pro) with 8 cores, M1 Apple Pro chip and $16$ GB RAM.

\section{ILLUSTRATIVE EXAMPLE OF ESCORT DISTRIBUTIONS}
\label{app:escortExample}
To illustrate how the escort version of a pdf makes the tails lighter with a concrete example, we show how the parameters of a Student-t distribution change when considering their escort version. In particular, the following proposition shows that it is possible to construct the escort pdf of a Student-t pdf such that the escort has a higher degree of freedom parameter than the original, and hence a lighter tail.

\begin{proposition}
    \label{prop:escortStudent-t} \citep{MLE-VI-lambdaExpFamilies}
    Consider two Student-t families in dimension $d$ with $\nu_q$ and $\nu$ degrees of freedom, respectively. 
    Then the escort $q_{\mu_q,\Sigma_q,\nu_q}^{(\alpha)}$ of $q_{\mu_q,\Sigma_q,\nu_q}$ with $\alpha = 1 + \frac{2}{\nu+d}$, is a Student-t distribution with $\nu^{(\alpha)}$ degrees of freedom, location $\mu^{(\alpha)}$, and shape $\Sigma^{(\alpha)}$ such that
    \begin{equation}
        \begin{cases}
            \nu^{(\alpha)} = \nu_q + 2 \frac{\nu_q + d}{\nu + d},\\
            \mu^{(\alpha)} = \mu_q,\\
            \Sigma^{(\alpha)} = \frac{\nu_q}{\nu^{(\alpha)}} \Sigma_q.
        \end{cases}
    \end{equation}
\end{proposition}

\section{DIVERGENCE AT THE OPTIMUM FOR SUTDENT-T TARGETS}
\label{app:wellPosednessStudentTarget}

We now study the properties of the optimization problem \eqref{eq:jointMinimization1} when the target is a Student-t distribution. In particular, we give the proof of Proposition \ref{eq:wellPosednessProblem}. We also describe in this case the inner problem in \eqref{eq:jointMinimization2} and give an explicit expression of its optimum value, leading to the plot in Figure \ref{fig:optimalRenyiValue}.

\begin{proof}[Proof of Proposition \ref{eq:wellPosednessProblem}]

The $\alpha$-divergence is such that $D_{\alpha}(p,q) \geq 0$ with equality if and only if $p = q$ almost everywhere (for $\alpha > 0$ and $\alpha \neq 1$). Moreover, for any $\nu > 0$, $\alpha(\nu) = 1 +\frac{2}{\nu+d} > 1$. This implies that if $(\mu,\Sigma,\nu)$ is such that
\begin{equation}
    D_{\alpha(\nu)}(\pi, q_{\mu,\Sigma,\nu}) = 0,
\end{equation}
then $(\mu,\Sigma,\nu)$ is a solution of Problem \eqref{eq:jointMinimization1}.

Since $\pi$ is a Student-t distribution, there exists $(\mu^{\star}, \Sigma^{\star}, \nu^{\star})$ such that $q_{\mu^{\star}, \Sigma^{\star},\nu^{\star}} = \pi$. In particular, $\nu^{\star} = \nu_{\pi}$. This implies that
\begin{equation}
    D_{\alpha(\nu^{\star})}(\pi, q_{\mu^{\star}, \Sigma^{\star},\nu^{\star}}) = 0,
\end{equation}
and hence the result.
\end{proof}

We now detail how to compute the function $\nu \longmapsto \mathrm{D}_{\alpha(\nu)}(\pi, q_{\mu^{\star}_{\nu},\Sigma^{\star}_{\nu},\nu})$ when $\pi$ is a Student-t distribution, as it is plotted in Figure \ref{fig:optimalRenyiValue}. To this end, we need to introduce the Rényi entropy of a pdf $p$ that is defined by
\begin{equation}
    H_{\alpha}(p) := \frac{1}{1-\alpha} \log \left( \int p(x)^{\alpha}dx \right).
\end{equation}
We then use this notion to give an explicit expression of our quantity of interest.

\begin{proposition}
    \label{prop:explicitAlphaDiv}
    Consider a target distribution $\pi$ and the family of Student-t distribution with $\nu$ degrees of freedom with $\alpha = 1 + \frac{2}{\nu+d}$. Consider $q_{\mu^{\star}_{\nu},\Sigma^{\star}_{\nu}, \nu}$ such that Equation \eqref{eq:optCondStudent} is satisfied. Then we have
    \begin{equation}
        \alpha(\alpha-1)\mathrm{D}_{\alpha}(\pi, q_{\mu^{\star}_{\nu},\Sigma^{\star}_{\nu}, \nu}) =  \exp \left( (\alpha-1) \left(H_{\alpha}(q_{\mu^{\star}_{\nu},\Sigma^{\star}_{\nu},\nu}) - H_{\alpha}(\pi) \right) \right) -1.
    \end{equation}
\end{proposition}

\begin{proof}
    We can see from Proposition \ref{prop:escortStudent-t} that Equation \eqref{eq:optCondStudent} implies that $\pi^{(\alpha)}(x) = q_{\mu^{\star}_{\nu},\Sigma^{\star}_{\nu},\nu}^{(\alpha)}(x)$ and $\pi^{(\alpha)}(xx^{\top}) = q_{\mu^{\star}_{\nu},\Sigma^{\star}_{\nu},\nu}^{(\alpha)}(x x^{\top})$. We can deduce from that, and using \cite[Equation (3.17)]{wong2022}, that the $\alpha$-divergence $\mathrm{RD}_{\alpha}$ between $\pi$ and $q_{\mu^{\star}_{\nu},\Sigma^{\star}_{\nu},\nu}$ is such that
    \begin{equation}
        \mathrm{RD}_{\alpha}(\pi, q_{\mu^{\star}_{\nu},\Sigma^{\star}_{\nu},\nu}) = H_{\alpha}(q_{\mu^{\star}_{\nu},\Sigma^{\star}_{\nu},\nu}) - H_{\alpha}(\pi),
    \end{equation}
    where $\mathrm{RD}_{\alpha}$ is the Rényi divergence with parameter $\alpha$. The result follows from the link between the $\alpha$-divergence and the Rényi divergence (see \citep{vanErven2014} for the definition and properties of the Rényi divergence).
\end{proof}

Proposition \ref{prop:explicitAlphaDiv} shows that, in order to compute the quantity plotted in Figure \ref{fig:optimalRenyiValue}, we need to compute explicitly the Rényi entropy of a Student-t distribution and compute explicitly the parameters of  $q_{\mu^{\star}_{\nu}, \Sigma^{\star}_{\nu}, \nu}$. We do so in the following two propositions.

\begin{proposition}
    Consider two degree of freedom parameters $\nu, \nu_q >0$, a dimension $d$, and set $\alpha = 1 + \frac{2}{\nu + d}$. Then, for any $q_{\mu_q,\Sigma_q,\nu_q}$, we have that
    \begin{equation}
        H_{\alpha}(q_{\mu_q,\Sigma_q,\nu_q}) = -\frac{\nu + d}{2} \left( \log Z_{\nu^{(\alpha)}, \Sigma^{(\alpha)}} - \alpha \log Z_{\nu_q, \Sigma_q} \right)
    \end{equation}
    with $\nu^{(\alpha)} = \nu_q + 2 \frac{\nu_q + d}{\nu + d}$ and $\Sigma^{(\alpha)} = \frac{\nu_q}{\nu^{(\alpha)}} \Sigma_q$.
\end{proposition}

\begin{proof}
    Using the result of Proposition \ref{prop:escortStudent-t}, we first compute that for any $x \in \mathbb{R}^d$,
    \begin{equation}
        q_{\mu_q, \Sigma_q, \nu_q}(x)^{\alpha} = \frac{Z_{\nu^{(\alpha)}, \Sigma^{(\alpha)}}}{Z_{\nu_q,\Sigma_q}^{\alpha}} q_{\mu^{(\alpha)}, \Sigma^{(\alpha)}, \nu^{(\alpha)}}(x)
    \end{equation}

    From there, we deduce that
    \begin{align*}
        H_{\alpha}(\mu_q, \Sigma_q, \nu_q) &= \frac{1}{1-\alpha} \log \left( \int q_{\mu_q, \Sigma_q, \nu_q}(x)^{\alpha}  dx  \right) \\
        &= \frac{1}{1-\alpha} \left( \log Z_{\nu^{(\alpha)}, \Sigma^{(\alpha)}} - \alpha \log Z_{\nu_q, \Sigma_q} \right)
    \end{align*}
    which gives the result.
\end{proof}

\begin{proposition}
    Consider two degree of freedom parameters $\nu, \nu_{\pi} >0$, a dimension $d$, and set $\alpha = 1 + \frac{2}{\nu + d}$. Then, for any $\pi = q_{\mu_{\pi},\Sigma_{\pi},\nu_{\pi}}$, the Student-t distribution $q_{\mu^{\star}_{\nu}, \Sigma^{\star}_{\nu}, \nu}$ minimizing $(\mu,\Sigma) \longmapsto \mathrm{D}_{\alpha}(\pi,q_{\mu,\Sigma,\nu})$ is such that
    \begin{equation}
        \begin{cases}
            \mu^{\star}_{\nu} = \mu_{\pi},\\
            \Sigma^{\star}_{\nu} = \frac{\nu_{\pi}}{\nu^{(\alpha)} - 2} \Sigma_{\pi},
        \end{cases}
    \end{equation}
    provided that $\nu^{(\alpha)} > 2$.
\end{proposition}

\begin{proof}
    This comes from the optimality result of Proposition \ref{prop:escortMM}, the characterization of ${\pi}^{(\alpha)}$ as a Student-t distribution with parameters given in Proposition \ref{prop:escortStudent-t}, and the fact that for any $q_{\mu,\Sigma,\nu}$ with $\nu>2$, $q_{\mu,\Sigma,\nu}(x) = \mu$ and $q_{\mu,\Sigma,\nu}(xx^{\top}) = \frac{\nu}{\nu-2}\Sigma$.
\end{proof}

Gathering these three results, we can then get a closed-form expression for the function 
\begin{equation}
    \nu \longmapsto \min_{\mu,\Sigma} \mathrm{D}_{\alpha(\nu)}(\pi, q_{\mu,\Sigma,\nu})
\end{equation}
when $\pi$ is Student-t distribution for some $\nu_{\pi} > 0$. Then, one can use it to draw Figure \ref{fig:optimalRenyiValue}.

\section{PROOFS OF \cref{sec:alphaess}}
\label{app:proof_ess}

We give below the proofs of Propositions \ref{prop:divergenceESS} and \ref{prop:clt}, that describe the approximation of $\alpha$-divergences by a self-normalized importance sampling estimator. This estimator is linked with the $\alpha$-ESS and the (discrete) $\alpha$-divergence between the normalized importance weights and the corresponding uniform weights.

\begin{proof}[Proof of \cref{prop:divergenceESS}]
We derive the following self-normalized IS (SNIS) approximation of the $\alpha$-divergence, making an ESS-like quantity appear:
\begin{align}
    \mathrm{D}_{\alpha}(\pi, q) =&  \frac{1}{\alpha (\alpha-1)} \left ( \frac{\int \left ( \frac{\widetilde{\pi}(x)}{q(x)} \right )^\alpha q(x) dx  } {Z_{\pi}^{\alpha}} -1 \right ) \label{eq:before_snis} \\ 
    \approx& \frac{1}{\alpha (\alpha-1)}  \left (\frac{\frac{1}{M} \sum_{m=1}^{M} \left (\frac{\widetilde{\pi}(x^{(m)})}{q(x^{(m)})} \right )^\alpha}{ \left ( \frac{1}{M} \sum_{m=1}^{M} \frac{\widetilde{\pi}(x^{(m)})}{q(x^{(m)})}  \right )^\alpha } - 1 \right )\label{eq:after_snis} \\ 
    =& \frac{1}{\alpha (\alpha-1)} \left ( M^{\alpha - 1}  \sum_{m=1}^{M} \overline{w}_{m}^\alpha - 1  \right )\label{eq:discreteAlphaDiv} \\ 
    =& \frac{ M^{ \alpha - 1}}{\alpha (\alpha - 1 )} \left ( \sum_{m=1}^{M} \overline{w}_{m}^\alpha -  M^{1- \alpha}  \right ) \\ 
    =&  \frac{ M^{ \alpha - 1}}{\alpha (\alpha - 1 )} \left ( (ESS_{\alpha})^{1 - \alpha} - M^{1 - \alpha} \right ) .
\end{align}

Notice the SNIS approximation from \cref{eq:before_snis} to \cref{eq:after_snis}, i.e., the same set of samples is used to approximate a ratio of two integrals. Moreover, we can recognize from \cref{eq:discreteAlphaDiv} that 
\begin{equation}
    \mathrm{D}_{\alpha}(\pi,q) \approx \mathrm{D}^M_{\alpha}(\{\overline{w}^{(m)}\}_{m=1}^M, \{ 1/M\}_{m=1}^M),
\end{equation}
with the continuous $\alpha$-divergence on the left and the discrete $\alpha$-divergence on the right. We also have the almost sure convergence $\mathrm{D}^M_{\alpha}(\{\overline{w}^{(m)}\}_{m=1}^M, \{ 1/M\}_{m=1}^M) \xrightarrow[M \rightarrow +\infty]{a.s.} \mathrm{D}_{\alpha}(\pi,q)$ from standard SNIS results (see for instance \cite[Theorem 9.2]{mcbook}).

\end{proof}

\begin{proof}[Proof of \cref{prop:clt}]
We first compute the gap
\begin{equation}
    D_{\alpha}^M(\{\overline{w}^{(m)}\}_{m=1}^M, \{ 1/M\}_{m=1}^M) - D_{\alpha}(\pi,q) = \frac{1}{\alpha (\alpha-1)} \left( \frac{\frac{1}{M} \sum_{m=1}^{M} \left (\frac{\widetilde{\pi}(x^{(m)})}{q(x^{(m)})} \right )^\alpha}{ \left ( \frac{1}{M} \sum_{m=1}^{M} \frac{\widetilde{\pi}(x^{(m)})}{q(x^{(m)})}  \right )^\alpha } - \frac{\int \left ( \frac{\widetilde{\pi}(x)}{q(x)} \right )^\alpha q(x) dx}{\left( \int \Tilde{\pi}(x)dx \right)^{\alpha}}  \right).
\end{equation}

We now deal with the denominator. Due to our hypothesis $\pi(x) > 0 \Rightarrow q(x) > 0$, we have that $\frac{1}{M} \sum_{m=1}^{M} \frac{\widetilde{\pi}(x^{(m)})}{q(x^{(m)})} \xrightarrow[M \rightarrow +\infty]{a.s.} Z_{\pi}$, from which we deduce the following almost sure convergence:
\begin{equation}
    \label{eq:almostSureProofCLT}
    \left(\frac{1}{M} \sum_{m=1}^{M} \frac{\widetilde{\pi}(x^{(m)})}{q(x^{(m)})}\right)^{\alpha} \xrightarrow[M \rightarrow +\infty]{a.s.} \left(\int \Tilde{\pi}(x)dx\right)^{\alpha}.
\end{equation}

We now turn to the numerator. The quantity $\frac{1}{M}\sum_{m=1}^M \left( \frac{\widetilde{\pi}(x^{(m)}}{q(x^{(m)})}\right)^{\alpha}$ is an unbiased Monte Carlo estimator of $\int \left ( \frac{\widetilde{\pi}(x)}{q(x)} \right )^\alpha q(x) dx$ with the variance of each term of the sum being equal to
\begin{equation}
    \mathbb{V}_q \left[ \left( \frac{\widetilde{\pi}(x)}{q(x)} \right)^{\alpha} \right] = \int \widetilde{\pi}(x)^{2\alpha} q(x)^{1-2\alpha} dx - \left( \int \left ( \frac{\widetilde{\pi}(x)}{q(x)} \right )^\alpha q(x) dx \right)^2.
\end{equation}
We have by the central limit theorem for Monte Carlo estimators that
\begin{equation}
    \label{eq:cltProofCLT}
    \sqrt{M} \left(\frac{1}{M}\sum_{m=1}^M \left( \frac{\widetilde{\pi}(x^{(m)}}{q(x^{(m)})}\right)^{\alpha} - \int \left ( \frac{\widetilde{\pi}(x)}{q(x)} \right )^\alpha q(x) dx \right) \xrightarrow[M \rightarrow +\infty]{d} \mathcal{N} \left(0, \mathbb{V}_q \left[ \left( \frac{\widetilde{\pi}(x)}{q(x)} \right)^{\alpha} \right] \right).
\end{equation}

Then, using Eq.~\eqref{eq:almostSureProofCLT}-\eqref{eq:cltProofCLT} and Slutsky's theorem, we obtain that
\begin{align}
    &\sqrt{M} \left( D_{\alpha}^M(\{\overline{w}^{(m)}\}_{m=1}^M, \{ 1/M\}_{m=1}^M) - D_{\alpha}(\pi,q) \right)\nonumber\\
    \xrightarrow[N\rightarrow +\infty]{d} &\:\mathcal{N}\left(0 , \left( \frac{\int \widetilde{\pi}(x)^{2\alpha} q(x)^{1-2\alpha}dx}{\left( \alpha(\alpha-1)\int \widetilde{\pi}(x)^{\alpha}q(x)^{1-\alpha}dx \right)^2} - 1 \right) \right),
\end{align}
which yields the result.
\end{proof}

\section{TAIL ADAPTATION}
\label{appendix:tailAdaptation}

\subsection{Our Tail Adaptation Procedure with Bayesian Optimization} 

We present below in more details our proposed tail adaptation procedure.

\begin{algorithm}[H]
    \caption{Tail adaptation with BO}
    \label{alg:tailadapt}
    \begin{algorithmic}[1]
    \REQUIRE \begin{itemize}
    \item Current tail parameter and $\alpha\text{-}\mathrm{ESS}$, i.e., $\nu_t,\widehat{\mathrm{ESS}}_{\alpha_{t}}$
    \item Previous tail parameters and $\alpha\text{-}\mathrm{ESS}$ values $\{ \nu_{\tau}, \widehat{\mathrm{ESS}}_{\alpha_{\tau}} \}_{\tau=1}^{t-1}$
    \item Choice of parameterized kernel function $k(\nu, \nu^\prime; \theta)$
    \item Choice of parameterized acquisition function $\mathrm{acq}(\nu; \psi; \mathcal{GP})$ 
    \item \textbf{(Optional):} Choice of prior distribution for $\theta$, $p(\theta)$
    \item \textbf{(Optional):} Choice of prior distribution for the observation noise $\sigma^2$, $p(\sigma^2)$
    \end{itemize}
    \STATE Use transformed values of $\widehat{\mathrm{ESS}}_{\alpha_{\tau}}$ with the following monotonic transformation 
    \begin{equation}\label{eq:monotonic_transform}
        y_{\tau} =  \log \left ( 1 - \left ( \frac{1}{M} \widehat{\mathrm{ESS}}_{\alpha_{\tau}} \right ) \right ) , ~~ \tau = 1,\dots, t
    \end{equation}
    \STATE Update Gaussian process $\mathcal{GP}_{t}$ at current iteration $t$ with new datapoint $\{ \nu_t, y_t \}$
    \STATE Obtain new tail parameter $\nu_{t+1}$ by maximizing the acquisition function, \begin{equation}\label{eq:maximizing_acquisition}
        \nu_{t+1} \gets  \argmax_{\nu} \mathrm{acq}(\nu; \psi; \mathcal{GP}_{t})
    \end{equation}
    \STATE \textbf{(Optional) Gaussian process hyperparameter optimization:} We model $\{ y_{\tau} \}_{\tau=1}^{t}$ as \emph{noisy} observations of the true (transformed) $\alpha$-ESS from the $\mathcal{GP}$ with additive Gaussian noise with variance $\sigma^2$.  Letting the observations be $y= [y_1, \dots, y_t]$ and $K_{t}$ be the $t \times t$ matrix with entries $k(\nu_{\tau}, \nu_{\tau^{\prime}}; \theta)$ for $(\tau,\tau^{\prime}) \in \{ 1, \dots, t \} \times  \{ 1, \dots, t \}$, optimize $\theta$ and noise $\sigma^2$ to maximize the log-likelihood of the observed data, as
    \begin{equation}\label{eq:hyperparam_optimization}
        \theta, \sigma^2 \gets \argmax_{\theta, \sigma^2} -\frac{1}{2} \log | \mathrm{det}\left [2 \pi\left( {K}_{t}+\sigma^2  {I}\right) \right ] | -\frac{1}{2}  {y}^{\top}\left( {K}_{t}+\sigma^2  {I}\right)^{-1}  {y} + \log p(\theta) + \log p(\sigma^2)
    \end{equation}
    \STATE \textbf{Return: $\nu_{t+1}$} 
\end{algorithmic}
\end{algorithm}
We describe below all implementation details regarding \cref{alg:tailadapt}.

\textbf{Kernel function.} We used the perhaps most common kernel function in BO, i.e., the radial basis function (RBF) kernel, also known as exponentiated quadratic (EQ) or squared exponential (SE) \citep{garnett_bayesoptbook_2023}. For a one dimensional input as $\nu$, the SE kernel has two scalar parameters, lengthscale $l$ and function variance $\sigma_f^{2}$, i.e., $\theta = \{ l ,\sigma_{f}^2 \}$, and its expression is given by 
\begin{equation}
    k_{\text{SE}}(\nu, \nu^\prime; \theta) = \sigma_{f}^2 \cdot \exp \left ( - \frac{1}{2} \frac{(\nu - \nu^\prime )^2}{l^2}  \right ) .
\end{equation}
The lengthscale $l$ indicates the typical distance between turning points in the function, while the intuition for $\sigma_f$ is that by seeing a long enough horizontal stretch of the function, $\approx 2/3$ of the points would lie between $\pm \sigma_f$ of the GP mean.

\textbf{Acquisition function.} We experimented using the Gaussian process upper confidence bxound (GP-UCB) \citep{srinivas2009gaussian}, a parameterized (by $\psi$) acquisition function, $\mathrm{acq}$, with only one scalar tuning parameter $\psi = \{\beta
\}, \beta > 0$ given by 
\begin{equation}\label{eq:ucb}
    \nu^{\text{UCB-best}}_{t+1} = \argmax_{\nu \in [1, \nu_{\text{max}}]} \mu_{\mathcal{GP}_t}(\nu) + \beta^{1/2} v_{\mathcal{GP}_t}^2(\nu) ,
\end{equation}
where, defining $k_t(\nu) $ as the vector-valued function $ \nu \rightarrow [k(\nu, \nu_{1}), \dots , k(\nu,\nu_t)]$ (and omitting kernel parameters $\theta$ for brevity), 
\begin{align}
     \mu_{\mathcal{GP}_t}(\nu) &= k_{t}(\nu)^\top (K_{t} + \sigma^2 I)^{-1} y_t ,\\
     v_{\mathcal{GP}_t}^2(\nu) &= k(\nu,\nu) - k_{t}(\nu)^\top (K_{t} + \sigma^2 I)^{-1}  k_{t}(\nu) . 
\end{align}
The $\beta$ parameter controls the typical exploration and exploitation tradeoff needed. To set $\beta$, we followed the theoretical guarantees described by \citep[Chapter 10, page 229]{garnett_bayesoptbook_2023}; letting the search space for $\nu$ be $\mathcal{V} = [1, \nu_{\text{max}}]$ and $t$ for the BO iteration number (corresponding to $t$ in our \ahtis{} algorithm), we selected 
\begin{equation}\label{eq:beta_selec}
    \beta_{t}^{\star} = \sqrt{2 \log \left(\frac{\left(t^2+1\right)|\mathcal{V}|}{\sqrt{2 \pi}}\right)} 
\end{equation}
for the synthetic experiments. For the real data experiments, we used $\beta_{t} = 1.5 \cdot \beta_{t}^{\star}$ for higher exploration due to a much noisier and more challenging objective function.
As search space for \cref{eq:ucb}, we used $\nu_{\text{max}} = 10$.

\textbf{Hyperparameter priors.} 
As described in the main paper, for the real data experiments we optimized the GP hyperparameters at each iteration (\textbf{step (4)} of \cref{alg:tailadapt}) using a prior for both $\theta = {l, \sigma_{f}^2}$ and $\sigma$. For all these parameters, we used an inverse Gamma prior, 
\begin{equation}
    p(\sigma^2 | \alpha, \beta) = \frac{\beta^\alpha}{\Gamma(\alpha)} (\sigma^2)^{-(\alpha + 1)} e^{-\frac{\beta}{\sigma^2}} , 
\end{equation}
where \( \Gamma(\cdot) \) is the gamma function, (omitting equivalent equations for $l$ and $\sigma_{f}^2$) with $\alpha$ and $\beta$ selected such that $\mathbb{E}[\sigma^2] = 3, \mathbb{V}[\sigma^2] = 2$; $\mathbb{E}[\sigma_{f}^2] = 5, \mathbb{V}[\sigma_{f}^2] = 2$; $\mathbb{E}[l] = 5, \mathbb{V}[l] = 2$.

Finally, to implement all of the above steps we used the Python library Emukit \citep{emukit2019,emukit2023}. 

\subsection{Another Tail Adaptation Method}
\cite{daudel2023monotonic} propose a VI method for minimizing a fixed $\alpha$-divergence over a mixture of Student-t distributions in \citep[Example 5]{daudel2023monotonic}. For each component of the mixture, the location, scale, and tail parameters are all adapted. We now show that their tail-adaptation procedure is not able to produce degree of freedom parameters that are less than a constant $\nu_{\textrm{min}} \in (2.5, 2.6)$.

In order to observe that, we consider the update \citep[Equation (70)]{daudel2023monotonic}. For simplicity, we consider the case where the mixture is reduced to one component, but our analysis still applies in this more general setting. In the simplified setting we consider, we have at iteration $t$ that the next degree of freedom parameter $\nu_{t+1}$ satisfies
\begin{equation}
    \label{eq:tailUpdateKamelia}
    \kappa \left(\frac{\nu_{t+1}}{2}\right) = \int (z - \ln(z)) p_{\mu_t, \Sigma_t,\nu_t}(y,z,dy,dz),
\end{equation}
with $\kappa(x) = \ln(x) + \frac{\Gamma'(x)}{\Gamma(x)}$, $\alpha \in [0,1)$, and a positive measure $p_{\mu_t, \Sigma_t,\nu_t}$ over $\mathbb{R} \times \mathbb{R}^d$. For any $z > 0$, we have $z - \ln(z) \geq 1$. We  can thus check that the right-hand side of Eq.~\eqref{eq:tailUpdateKamelia} is positive. The function $\kappa$ is increasing and bijective from $(0,+\infty)$ to $\mathbb{R}$ from \cite[Lemma 13]{daudel2023monotonic}.

Now let us demonstrate that there exists a scalar $\nu_{\textrm{min}} > 0$ such that $\nu_{t+1} > \nu_{\textrm{min}}$ and give some bounds on $\nu_{\textrm{min}}$. We define $\nu_{\textrm{min}}$ such that $\kappa \left( \frac{\nu_{\textrm{min}}}{2} \right) = 0$. The function $\kappa$ is increasing and bijective from $(0,+\infty)$ to $\mathbb{R}$ from \cite[Lemma 13]{daudel2023monotonic}. We can check that that $\kappa(1.25) < 0$ and that $\kappa(1.3) > 0$. This means that the scalar $\nu_{\textrm{min}}$ exists and satisfies $\nu_{\textrm{min}} \in (2.5, 2.6)$. This shows that there are values of $\nu$ that cannot be attained by the algorithm of \cite{daudel2023monotonic}. Although this lower bound is reasonable, it may not yield optimal performance on heavy-tailed targets such as the one considered in Section \ref{subsec:studentTargets}.

\section{FURTHER NUMERICAL EXPERIMENTS}
\label{appendix:furtherNumericalExperiments}

\subsection{Controlled Scenario with Varying Dimension Student-t Targets}

We give here supplementary numerical experiments in the case of a Student-t target distribution in varying dimension, that is described in \cref{subsec:studentTargets}. In addition to the results already presented in \cref{subsec:studentTargets}, we show in \cref{fig:StudentTarget_5dofTarget} the $\alpha$-ESS and square-root relative MSE on the normalization constant of the target when the target has degree of freedom $\nu_{\pi} = 5$. We also describe the final degree of freedom parameters reached by \ahtis{} with adaptation of $\nu$ when the target has degree of freedom parameter $\nu_{\pi} \in \{2,5\}$ in Table \ref{tab:dof_adaptation}.

\begin{figure}[H]
      \centering
      \begin{subfigure}[t]{.49\linewidth}
      \centering
          \includegraphics[width=\linewidth]{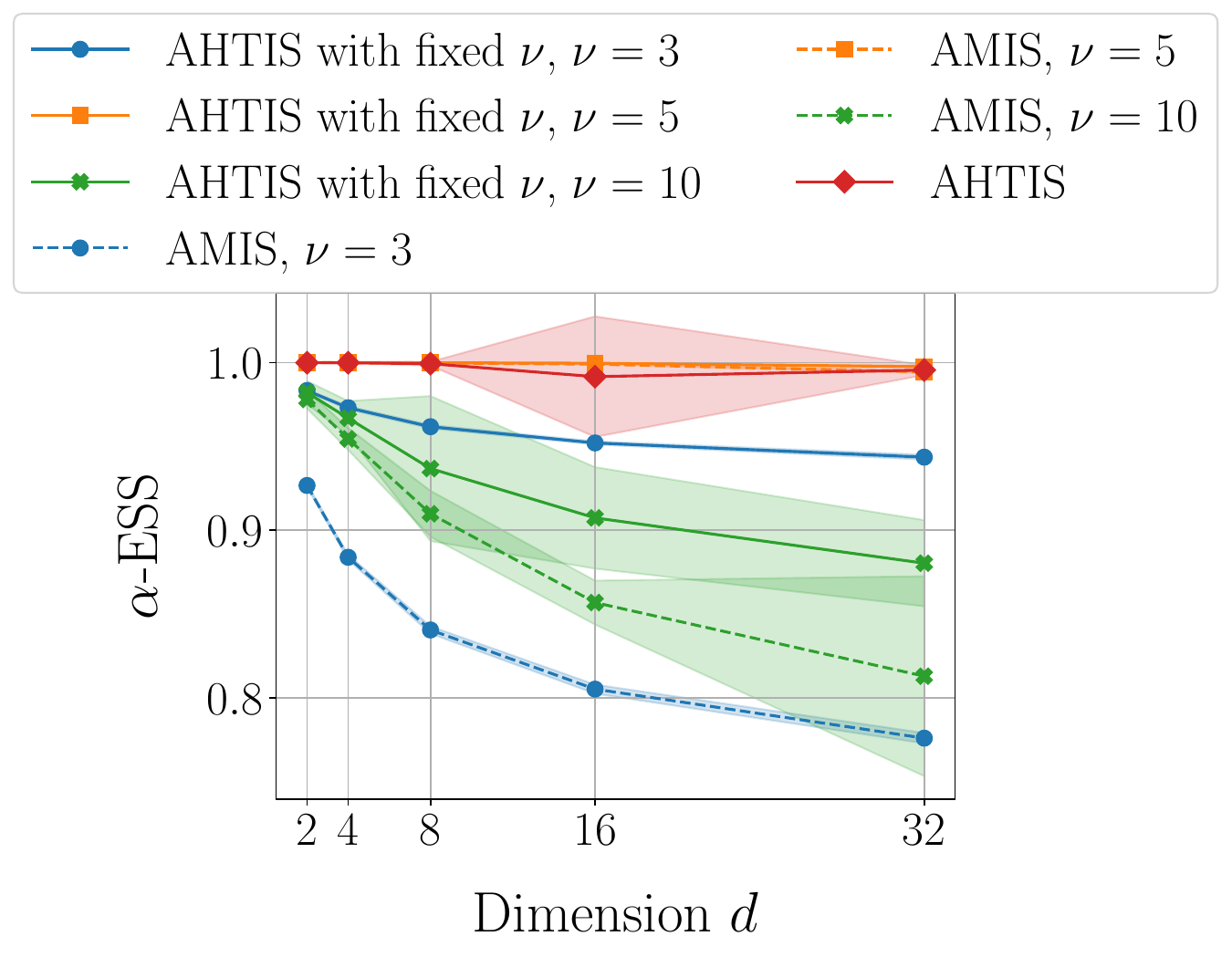}
            \caption{$\alpha$-ESS (mean $\pm$ one standard deviation, higher is better) for various dimensions $d$. \ahtis{} outperforms AMIS for any $\nu$, sometimes by an order of magnitude, and the $\nu$-adaptive version converges to the true value $\nu_{\pi} = 5$.}
            \label{fig:alphaESS_5dofTarget}
      \end{subfigure} \hfill 
      \begin{subfigure}[t]{.49\linewidth}
      \centering
          \includegraphics[width=\linewidth]{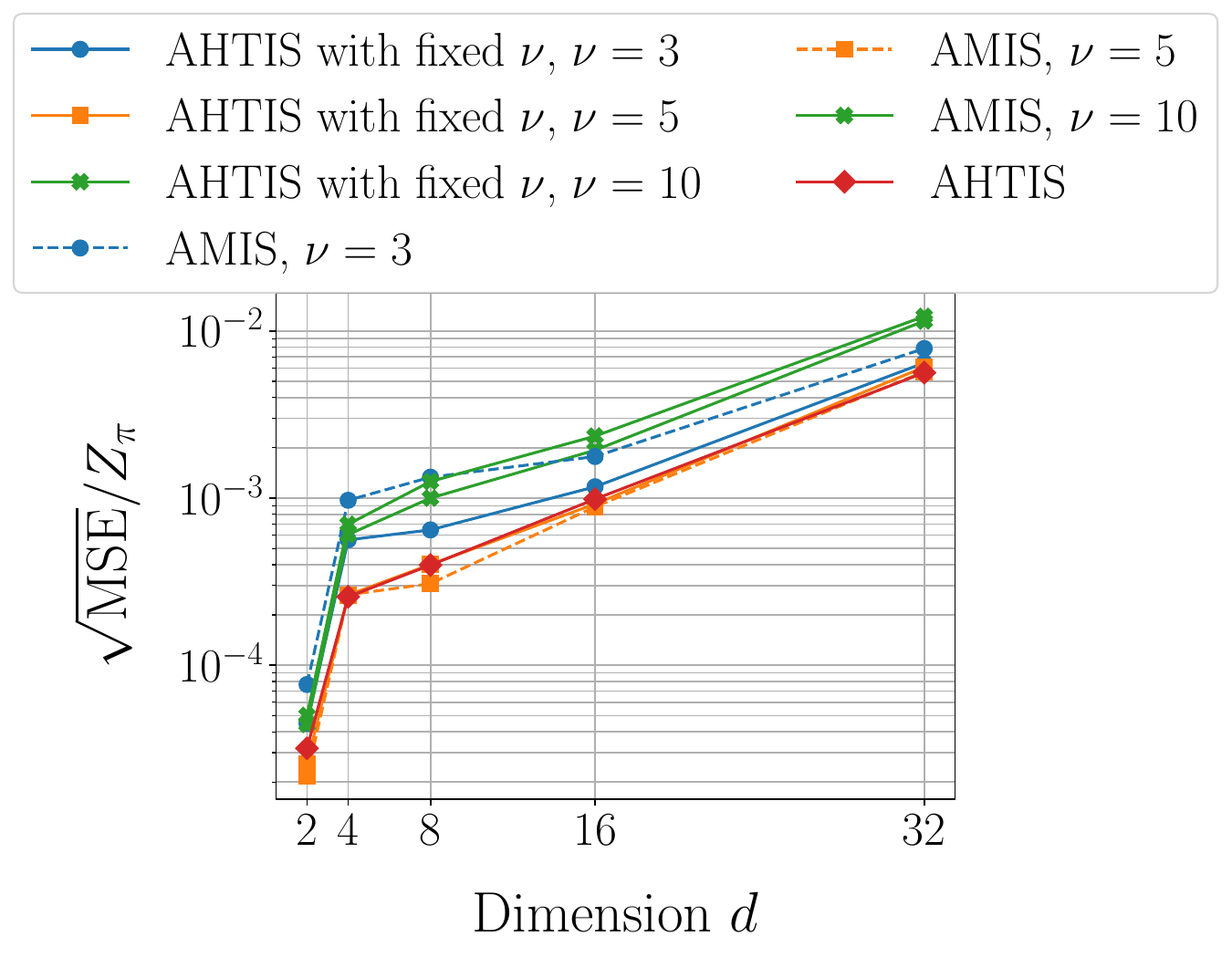}
                    \caption{Relative square root MSE (lower is better) for various dimensions $d$. Note that $Z_{\pi} = Z_{\nu_{\pi},\Sigma_{\pi}}$ is the true normalizing constant, which is available. \ahtis{} outperforms AMIS for any $\nu$ and the $\nu$-adaptive version converges to $\nu_{\pi} = 5$.}
            \label{fig:MSE_Z_2dofTarget}
      \end{subfigure}

\caption{Results for a synthetic Student-t target with $\nu_{\pi} = 5$. All algorithms are run for $T = 20$ iterations, with $M=10^4$ samples per iteration and results are averaged over $100$ replications. A dashed line identifies AMIS, while solid line is \ahtis{}, and same marker/color indicates same $\nu$.}
\label{fig:StudentTarget_5dofTarget}
\end{figure}

\begin{table}[H]
    \centering
    \begin{tabular}{c|cc}
               & $\nu_{\pi} = 2$        & $\nu_{\pi}=5$ \\
        \hline
        $d=2$  & $2.05 \pm 0.569$  & $4.87 \pm 0.142$ \\
        $d=4$  & $2.16 \pm 1.18$  & $4.96 \pm 0.168$\\
        $d=8$  & $1.98 \pm 0.574$ & $4.93 \pm 0.293$\\
        $d=16$ & $2.03 \pm 0.562$ & $5.03 \pm 0.978$\\
        $d=32$ & $2.04 \pm 0.175$ & $5.03 \pm 0.504$
    \end{tabular}
    \caption{Final degree of freedom $\nu_T$ yielded by Algorithm \ref{alg:ours} for a target with degree of freedom parameter $\nu_{\pi}$ in dimension $d$. The results are of the form mean $\pm$ one standard deviation, for $T=20$ over $100$ runs.}
    \label{tab:dof_adaptation}
\end{table}

\paragraph{Results.} Table \ref{tab:dof_adaptation} reveals that the $\nu$-adaptive \ahtis{} is able to correctly capture the tail behaviour of the target with good precision. Fig.~\ref{fig:StudentTarget_5dofTarget} shows a situation where AMIS and $\ahtis$ with $\nu=5$, and the $\nu$-adaptive \ahtis{} are able to reach similar performance in terms of $\alpha$-ESS and MSE. The fact that AMIS is now able to reach performance similar to \ahtis{} (in contrast with the results of Fig.~\ref{fig:synthetic_res}) is because $\nu_{\pi} = 5$, meaning that $\pi$ has well-defined first and second order moments and that AMIS with $\nu =\nu_{\pi}$ can be used. Note however that in the case of a mismatch $\nu \neq \nu_{\pi}$, AMIS is inferior to \ahtis{}.

\subsection{Application to Bayesian Student-t Regression on Real Data}
We include figures with added results for $\nu=4$ for all algorithms (excluded from the main paper for better readability of the main plots). 

\begin{figure}[H]
      \centering
      \begin{subfigure}[t]{.49\linewidth}
      \centering
          \includegraphics[width=\linewidth]{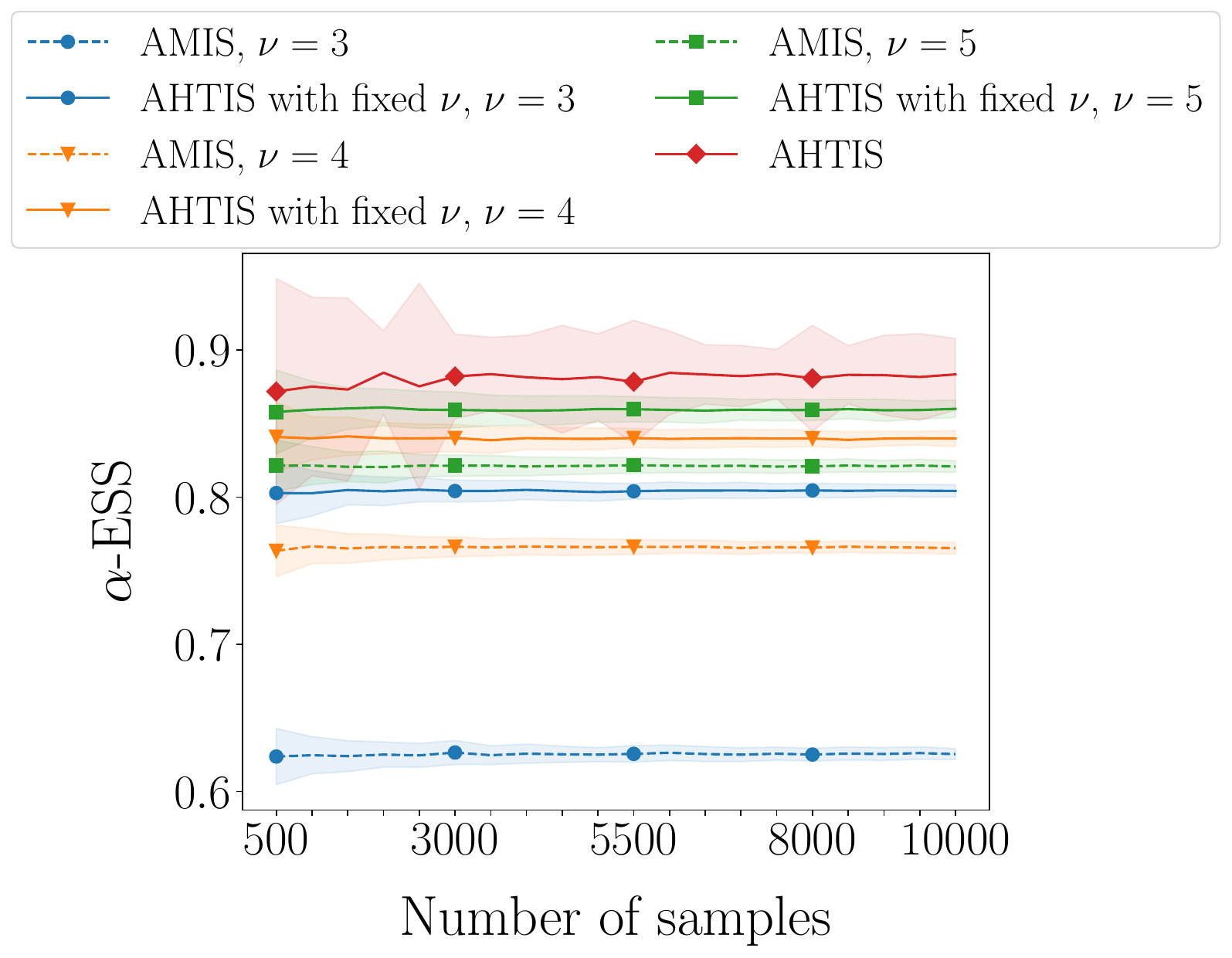}
            \caption{$\alpha$-ESS (mean $\pm$ one standard deviation, higher is better) for the creatinine dataset experiments. }
            \label{fig:1}
      \end{subfigure} \hfill 
      \begin{subfigure}[t]{.49\linewidth}
      \centering
          \includegraphics[width=\linewidth]{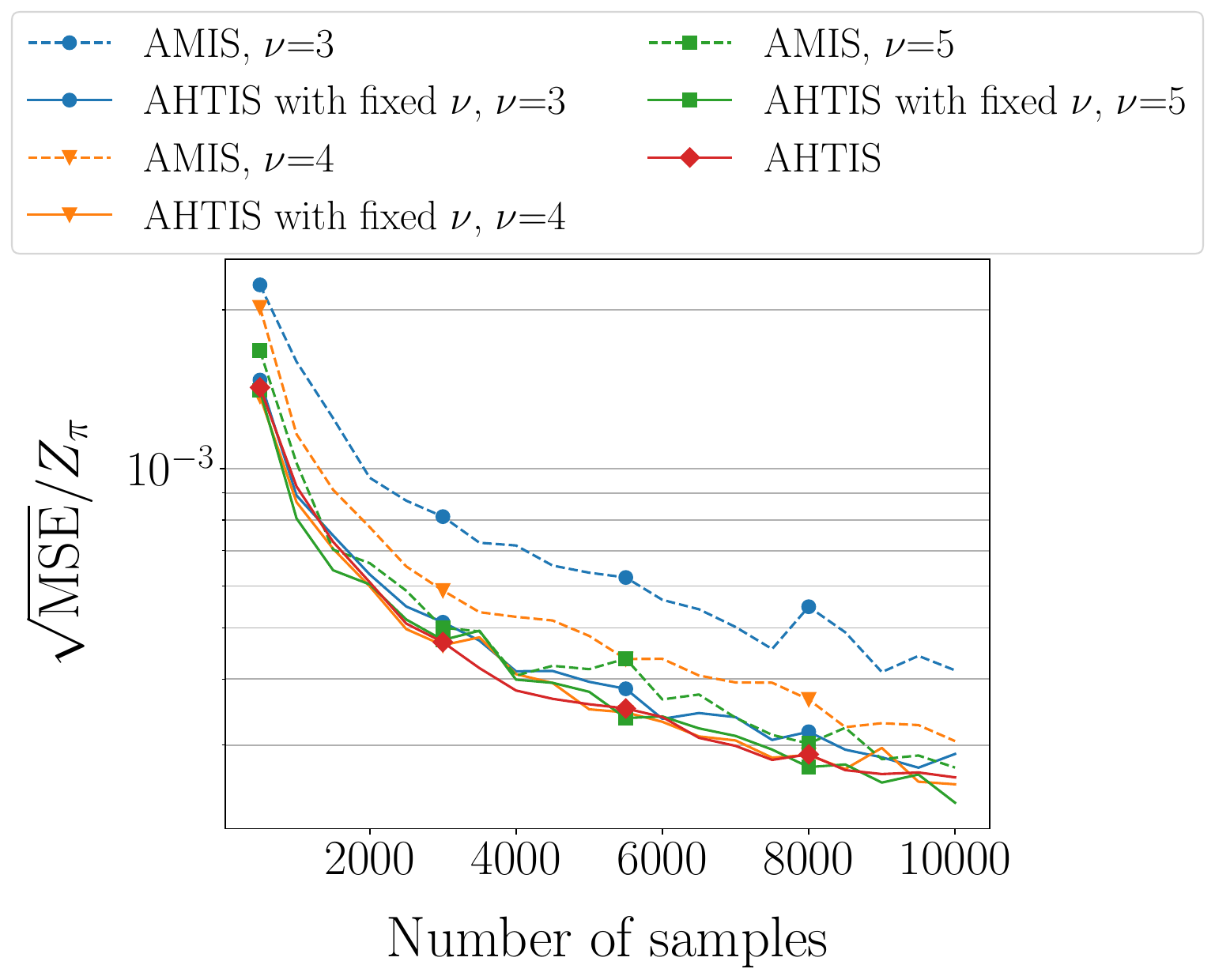}
                    \caption{Relative square root MSE (lower is better) for the creatinine dataset experiments. }
            \label{fig:2}
      \end{subfigure}

\caption{Results here are as in \cref{fig:alphaESS_real,fig:MSE_Z_real}, but with added $\nu=4$. Recall that all algorithms are run for $T = 25$ iteration and results are averaged over $250$ replications. A dashed line identifies AMIS, while solid line is \ahtis{}, and same marker/color indicates same $\nu$.}
\label{fig:further_results}
\end{figure}

\end{document}